\documentclass[12pt]{article}

\usepackage[utf8]{inputenc}

\usepackage[margin=1in]{geometry}
\usepackage{setspace}

\usepackage{amsmath, amssymb, amsthm}
\usepackage{bbm}
\usepackage{dsfont}

\usepackage{graphicx}
\usepackage{booktabs}
\usepackage{multirow, array}
\usepackage{float}

\usepackage{chngcntr}
\counterwithout{figure}{section}
\counterwithout{table}{section}
\counterwithout{table}{subsection}

\usepackage[title]{appendix}

\usepackage[round]{natbib}

\usepackage[hidelinks]{hyperref}
\newtheorem{proposition}{Proposition}
\newtheorem{assumption}{Assumption}

\title{Empirical Global Games of Regime Change}

\author{Matthew J. Baker \\ Department of Economics \\ Hunter College and the Graduate Center, CUNY \\ \\ Khaled Eltokhy \\ Department of Economics \\ The Graduate Center, CUNY\\  \\ Weichao Guo \\ School of Economics \\ Fudan University \\ }

\begin{document}

\maketitle 

\begin{abstract}
Global games theory provides a tractable framework for analyzing coordination problems with multiple equilibria, with regime overthrow serving as a canonical application. A large empirical literature on coups d'\'etat examines the relationship between country-level characteristics, coup occurrence, and coup success using reduced-form approaches that leave the underlying coordination problem implicit. Bridging these literatures, we develop an estimable global games model of coups d'\'etat. The model incorporates strategic coordination into the empirical analysis of coups, employing the global games framework as an equilibrium selection device. The model distinguishes between the feasibility and desirability of regime overthrow, allowing observable fundamentals to enter separately into beliefs about regime strength and perceived gains from rebellion. The model therefore provides a theoretical basis for decomposing coup outcomes into feasibility and desirability components under maintained exclusion restrictions and equilibrium assumptions. If information on coup strength is available, the model also allows estimation of agents' uncertainty about regime strength. We show how the model can be estimated using simulated maximum likelihood coupled with a contraction mapping, and demonstrate how observable covariates map into regime strength and the perceived benefits of overthrow. We illustrate how the framework can be applied through counterfactuals varying coup benefits, regime strength, and information quality. 
\end{abstract}

\thispagestyle{empty}

\newpage

\clearpage
\pagenumbering{arabic}

\onehalfspacing

\section{Introduction}

Coordination games with multiple equilibria play a central role in the analysis of a wide range of economic phenomena. The global games framework introduced by \citet{carlsson1993} has become one of the principal tools for analyzing such environments and is now a cornerstone of modern economic theory. The canonical application, described by \citet{angeletos2007a} and referred to as the ``workhorse of the applied literature'' by \citet[p. 5328]{morrispal}, is a game of regime change: a mass of citizens simultaneously decide whether to participate in an attempt to overthrow the governing regime in their country. Because the likelihood of successful overthrow increases with the number of participants, participation decisions are complementary, and the game may admit multiple equilibria. The global games approach resolves this indeterminacy by introducing incomplete information about regime strength, yielding a unique equilibrium characterized by cutoff strategies.

We transform this canonical global game of regime change into a structural econometric model of coups d'\'{e}tat. Modeling coups in this way makes strategic coordination explicit and addresses identification and selection issues that arise in standard empirical analyses of coup occurrence and success. The framework we develop distinguishes between the desirability of overthrow, which captures the incentives of citizens to participate, and the feasibility of overthrow, which reflects regime strength. Observable covariates are allowed to enter these channels separately, generating a structural decomposition of coup risk and success. Whereas existing empirical work typically relies on reduced-form or multi-equation approaches \citep[e.g.,][]{powell2012}, our approach links observed coup attempts and outcomes directly to strategic incentives and coordination constraints. 

The paper contributes to the literature on structurally estimating game-theoretic models \citep{Bajari_Hong_Nekipelov_2013}. Global games are particularly well suited to empirical implementation because incomplete information transforms a coordination problem with potential multiple equilibria into one with a unique equilibrium, thereby generating a robust and estimable mapping from model parameters to observed outcomes. Following \citet{tamer03}, we view this coherence as a key virtue of the framework from an econometric standpoint.

The remainder of the paper proceeds as follows. Section \ref{litreview} reviews the theoretical global games literature and the empirical literature on coups d'\'{e}tat. Section \ref{theorysec} develops the empirical version of the model. Section \ref{empact} presents estimation results for two specifications: a limiting case in which information about regime strength becomes arbitrarily precise and a more general model that estimates informativeness directly. Section \ref{app} presents some simple counterfactual exercises, and Section \ref{conc} concludes.

\section{Literature} \label{litreview}

Our work is motivated by two complementary ideas. The first is that the theory of global games has emerged as one of the most influential approaches for analyzing coordination problems with multiple equilibria, as it provides a tractable and robust framework for equilibrium selection in environments characterized by strategic complementarities and incomplete information. The second is that questions of regime stability and regime change are among the most important coordination problems in political economy, and while some theory has studied these problems from a global games perspective, the empirical implications of adopting this perspective on regime change are not clear. This paper seeks to bring these two strands of research together by applying the insights of global games to the empirical analysis of coups and regime change.

The global games approach demonstrated that even small amounts of incomplete information can fundamentally alter coordination outcomes, yielding sharp comparative statics and transforming coordination games from exercises in equilibrium selection into models with empirically meaningful predictions. The foundational contributions by \citet{carlsson1993} and \citet{morris2003} established the global-games framework as a tractable approach to equilibrium selection under incomplete information, highlighting how collective outcomes depend jointly on underlying fundamentals and agents' beliefs about the actions of others. Much subsequent work has applied these ideas to political transitions and regime change, emphasizing the role of information, beliefs, and strategic complementarities in collective action \citep{angeletos2007a,edmond13}. More broadly, the political economy literature has analyzed regime instability, democratization, and political conflict as the outcome of strategic interactions between governments and citizens. For example, \citet{acemoglu2001extension} develop a theory of political transitions based on the distribution of political power, while \citet{acemoglu2006economic} provide a broader framework linking economic conditions, political institutions, and regime change. Related work examines the incentives facing military elites and the persistence of political institutions, highlighting mechanisms through which regimes may remain vulnerable to overthrow or resist political change \citep{acemoglu2010revolution,acemoglu2008income}.

A more specialized branch of this literature applies coordination-based frameworks directly to understanding the driving forces behind coups, protests, elections, and autocratic instability. \citet{Little2017} develops a model of coups as a coordination problem under incomplete information, while \citet{Little2015,little2016,LittleEtAl2015}, \citet{hollyer2015}, are among those who have examined how information, communication, and transparency affect collective political action. More recent contributions have extended these ideas to protests, coups, and regime change, emphasizing the role of strategic uncertainty in political instability \citep{casper2014,TysonSmith2018,bueno22,morshed23,kocak24}. Collectively, these studies illustrate how coordination-based approaches have been applied to a wide range of political environments in which collective action depends on both fundamentals and beliefs.

While the aforementioned theoretical literature emphasizes strategic coordination, information, and beliefs, the empirical literature on coups has developed more or less independently. Most empirical work focuses on identifying the political, economic, and institutional correlates of coup activity using cross-national data, with comparatively little attention paid to the coordination mechanisms emphasized in theoretical models. Rather than modeling strategic interaction directly, this literature has generally sought to identify the observable political, economic, and institutional conditions associated with coup activity.

The quantitative analysis of coups includes early quantitative contributions by \citet{jackman1978}, \citet{muller1970}, and \citet{muller1985}, who examined the political, economic, and institutional determinants of military intervention and political instability using cross-national data. The modern empirical literature on coups is often traced to the influential contribution of \citet{londregan1990}, which developed a systematic political-economy framework for analyzing coup occurrence. More recent work has assembled cross-national panel data and developed empirical frameworks for studying both coup incidence and coup outcomes. For example, \citet{BelkinSchofer2003} examines the institutional and political determinants of coups, while \citet{PowellThyne2011} provides a comprehensive cross-national dataset that has become a standard resource in the field. More recent research has benefited from improved data collection and increasingly sophisticated empirical methods while remaining focused on the observable correlates of coup activity \citep{chin21, cebotari_political_2024}.

One strand of this literature focuses primarily on the determinants of coup occurrence. Studies have emphasized political instability, institutional weakness, regime durability, elite competition, economic performance, and military organization as important predictors of coup activity \citep{JenkinsKposowa1992,BelkinSchofer2003,BjornskovRode2020,BodeaElbadawiHoule2017}. A related literature examines coup-proofing and other strategies through which regimes seek to reduce coup risk, including military organization, patronage, repression, foreign assistance, and information management \citep{quinlivan99,Leon2014,Sudduth2017,TysonSmith2018,Masaki2016,AllenEtAl2022,ChiozzaKhalifa2024}. Although there are many variations and no single specification has emerged as standard, these studies consistently emphasize the importance of economic conditions, political institutions, and regime strategies in shaping the likelihood of military intervention.

Although some work that uses a simultaneous equations approach to studying coup occurrence and success \citep{powell2012}, a large literature focuses its attention on the determinants of coup success conditional on an attempt. While studies of coup incidence focus primarily on the incentives for rebellion, this literature emphasizes the feasibility of overthrow and the factors that determine whether collective action succeeds once initiated. Historical and comparative studies stress the importance of institutional context, military organization, and long-run political development. For example, \citet{OKane1981} highlights the role of state structures in shaping coup outcomes, while \citet{Darkwa2023}, \citet{AbiddeKumahAbiwu2023}, and \citet{ChambruManeuvrierHervieu2024} examine how political institutions and historical legacies influence regime resilience. Although primarily theoretical in orientation, \citet{Little2017} further emphasizes how beliefs about regime durability and the actions of others can shape the success or failure of attempts at regime overthrow. While this literature has generated important insights into both coup attempts and coup success, it generally lacks a unified framework that explains these outcomes as jointly determined by the same strategic environment. Consequently, the relationship between the determinants of coup occurrence and those of coup success remains somewhat theoretically underdeveloped.

Despite their common focus on regime overthrow, the theoretical literature on coordination and the empirical literature on coups have developed largely independently. Theoretical models emphasize strategic complementarities, information, beliefs, and equilibrium selection, while empirical studies typically estimate reduced-form relationships between observable covariates and coup outcomes. Empirical applications of global-games models remain comparatively rare and have generally focused on deriving testable implications rather than estimating the underlying strategic environment directly.\footnote{Empirical work directly evaluating global-game predictions has largely taken the form of laboratory experiments. For example, \citet{HEINEMANN2024632} examines whether global-game equilibrium selection predicts behavior in coordination environments and finds only limited support, while \citet{helland21} shows that boundedly rational alternatives based on \textit{k}-level reasoning often better explain observed behavior.} As a result, the mechanisms that generate coup attempts and determine their success are often treated implicitly rather than modeled directly.

We seek to bridge these gaps in this paper. Regime overthrow is arguably the canonical application of global games, yet most applications to political instability remain theoretical. We exploit the equilibrium structure of a global game to construct an estimable model of coup dynamics.\footnote{A small empirical literature has applied coordination-based methods to settings such as bank runs, debt runs, and currency attacks. For example, \citet{schroth2014} structurally estimates a dynamic debt-run model using evidence from the asset-backed commercial paper crisis, while \citet{cipriani2024} use high-frequency payment data to identify bank runs and distinguish the roles of fundamentals and coordination. Although these settings share important strategic features with the present paper, they focus on different empirical environments and do not attempt to estimate a global-games model of regime overthrow.} The framework decomposes coup activity into two latent components: the desirability of overthrow and the feasibility of overthrow. Existing empirical approaches typically combine these channels into reduced-form estimates of coup occurrence and success. By separating them structurally, the model provides a way to distinguish changing incentives for rebellion from changes in regime resilience. In doing so, the paper contributes to the literature on structural estimation of strategic interaction \citep{Bajari_Hong_Nekipelov_2013} and to the broader econometric literature on identification in latent-variable and selection models \citep{heckman78,vytlacil2002,matzkin2007,tamer03}. Unlike many empirical games with multiple equilibria, however, the global-games framework provides equilibrium selection directly \citep{carlsson1993,morris2003}, yielding coherent empirical predictions and a tractable basis for estimation.

This distinction is particularly relevant given long-run trends in coup dynamics, a fuller discussion of which we reserve for Section \ref{sec:est}. Although the frequency of coup attempts has declined over time, conditional success rates have remained comparatively stable. Existing empirical approaches provide limited guidance for interpreting this pattern: fewer coups may reflect stronger regimes, weaker incentives for rebellion, changing information environments, or shifts in coordination technology. Our framework provides a way to distinguish among these competing explanations. We now turn to the canonical model of regime overthrow and how we adapt it for empirical implementation.

\section{Theory} \label{theorysec}

To maintain close correspondence with the cross-country panel literature on coups d`\'{e}tat and with our empirical analysis, we model each country-year as an independent two-stage game, generating a panel of strategic interactions across countries $c$ and time periods $t$. At each place and time, the game unfolds in two stages: a desirability stage followed by a participation stage. This game is played by the mass of citizens in the country at the time.\footnote{To simplify the exposition, we refer to potential coup participants as ``citizens.'' The analysis, however, applies more generally to any collection of agents with an interest in regime change, including military officers, political elites, organized opposition groups, or members of the broader population. The model should therefore not be strictly interpreted as requiring mass participation by society as a whole.} In the initial desirability stage, citizens observe the mean benefits $b_{ct}$ that would accrue to participants if the incumbent regime were overthrown in a coup d'\'{e}tat. If $b_{ct}\leq 0$, a coup is not beneficial to the potential participants, no coup occurs, and the game ends.

If $b_{ct}>0$ a coup attempt occurs. In our framework, this means that the game proceeds to the participation stage, at which point the game takes on the typical characteristics of a global game, following the canonical setup described in \citet{angeletos2007a}. The condition $b_{ct}>0$ can be viewed as a necessary condition for a global game of regime change to take place at all, and as a way of modeling whether a coordination game of regime overthrow materializes. 

The ensuing global game setup is as usual. Each citizen in the unit mass of citizens indexed by $i \in [0,1]$, must choose whether to participate in the coup d'\'{e}tat,  where $a_{ict}\in\{0,1\}$ denotes the choice of each citizen to not participate ($a_{ict}=0$) or participate ($a_{ict}=1$). The total mass of citizens participating is $A_{ct}=\int_0^1a_{ict}di$, and the success of the coup depends on the size of the coup relative to the strength of the regime, which is captured by a state variable $\theta_{ct} \in (-\infty,\infty)$. If $A_{ct}\geq \theta_{ct}$, the coup is successful, so larger values of $\theta_{ct}$ correspond with stronger regimes. It follows that some coups are inevitable ($\theta_{ct}\leq 0$), some are impossible ($\theta_{ct}>1$), and some are only possible with sufficient participation ($\theta_{ct} \in (0,1))$. Participation in a successful coup yields participants the aforementioned benefit $b_{ct}$.\footnote{One may interpret $b_{ct}$ as the mean benefit of a successful coup, with individual-specific returns realized only after the coup succeeds if one wishes to allow for heterogeneity in payoffs.} If the coup fails, the participating citizen receives zero benefits and instead incurs cost $k_{ct}$.\footnote{The payoffs to an agent therefore have an expected utility structure of the form $Pb-(1-P)k$, where $P$ is the probability of a successful coup, $b$ are benefits, and $k$ are costs incurred when the coup fails. This setup has been applied to understanding coups in the applied literature as in, for example \citet{gassebner2016expect}.}  Payoffs for agent $j$ in country $c$ at time $t$ from participating in a coup given that $b_{ct}>0$ can then be written as:

\begin{equation} \label{payoffs}
u_{jct}
=
a_{jct}
\left[
b_{ct}\mathbf{1}(A_{ct}\geq\theta_{ct})
-k_{ct}\mathbf{1}(A_{ct}<\theta_{ct})
\right] , \quad 
A_{ct}
=
\int a_{ict}\,di
\end{equation}

Introducing uncertainty about regime strength $\theta_{ct}$ yields the familiar global games environment in which equilibrium behavior is characterized by a cutoff rule. We assume variability in regime strength and uncertainty about it follows the normal-normal structure commonly employed in the global games literature (see, for example, \citet{morris2003} and \citet{angeletos2007a}), as described in Assumption \ref{nnmod}:

\begin{assumption}[Normal--Normal Model] \label{nnmod}
The strength of the regime, $\theta_{ct}$, is distributed according to
\[
\theta_{ct}\sim N(z_{ct},\tau_{ct}^{2}),
\]
where $z_{ct}$ denotes the mean regime strength and $\tau_{ct}>0$ its standard deviation.

Each citizen $i$ observes an i.i.d. private signal of realized regime strength given by
\[
x_{ict}\sim N(\theta_{ct},\sigma_{ct}^{2}),
\]
where $\sigma_{ct}$ denotes the standard deviation of the idiosyncratic signal noise. 
\end{assumption}

We discuss this point further in Section \ref{empact}, but the normal-normal specification is chosen not merely for analytical tractability. More importantly, it induces latent index representations that can be exploited for estimation in much the same way as in standard discrete choice models. The Bayes-Nash equilibrium strategy of this game is for citizens to follow a common cut-off strategy: participate in the coup and attack the regime if $x_{ict}\leq x_{ct}^*$. \citet{morris2003} show this result and in addition describe the conditions under which the strategy is unique.\footnote{The key assumptions \citep[p. 65]{morris2003} are: payoffs must be monotone in the average action and the state parameter $\theta$, there must be a unique $\theta$ at which payoffs are zero for any average action, there must be threshold values of $\theta$ for which participating and not participating are dominant strategies, and payoffs must be continuous.} While all the usual conditions apply here, the requirement that the signal standard deviation $\sigma_{ct}$ be sufficiently small relative to $\tau_{ct}$ plays a special and nuanced role in our empirical implementation, so we shall discuss this assumption further below. 

We characterize equilibrium in the usual way \citep{angeletos2007a, edmond13}. Attack size conditional on a regime strength $\theta_{ct}$ is determined by the fraction of agents who receive a signal lower than the cutoff signal $x_{ct}^*$. Thus, conditional on $\theta_{ct}$ and a decision rule, attack size is:
\begin{equation} \label{attack}
    A_{ct} = \Phi\left[\frac{x_{ct}^* - \theta_{ct}}{\sigma_{ct}}\right]
\end{equation}

Where $\Phi[Z]$ is the cumulative standard normal function. If $A_{ct} \geq \theta_{ct}$, the regime is overthrown, while if $A_{ct} <\theta_{ct}$, the regime survives. As the right-hand side of equation (\ref{attack}) is strictly decreasing in $\theta_{ct}$, there exists a unique cutoff value of $\theta^*_{ct}$ given $x_{ct}^*$: 
\begin{equation} \label{attack_fail}
    \Phi\left[\frac{x_{ct}^* - \theta_{ct}^*}{\sigma_{ct}}\right]=\theta^*_{ct}
\end{equation}

The critical value $\theta_{ct}^*$ divides values of $\theta_{ct}$ for which the regime survives ($\theta_{ct}>\theta_{ct}^*$) from values for which it does not ($\theta_{ct}\leq\theta^*_{ct})$. 

From the perspective of the individual citizen, this reduces the problem of predicting the actions and information of the mass of other citizens to an assessment of the value of $\theta_{ct}$ given his or her information $x_{ict}$. Given Assumption~\ref{nnmod}, the citizen updates his or her beliefs about regime strength $\theta_{ct}$ upon receiving signal $x_{ict}$ in standard Bayesian fashion, so that his or her beliefs about $\theta_{ct}$ are characterized by the posterior distribution:

\begin{equation} \label{conddist}
    \theta_{ct}|x_{ict} \sim N\left(\frac{\sigma_{ct}^2z_{ct} + \tau_{ct}^2x_{ict}}{\sigma_{ct}^2+\tau_{ct}^2},\frac{\sigma_{ct}^2\tau_{ct}^2}{\sigma_{ct}^2+\tau_{ct}^2} \right)
\end{equation}

Following the distribution in (\ref{conddist}), a citizen receiving signal $x_{ict}$ has expectation of the likelihood that a coup is successful of:

\begin{equation} \label{expprob}
    P(\theta_{ct}\leq\theta_{ct}^*|x_{ict}) = \Phi\left[\frac{\theta_{ct}^*-\frac{\sigma_{ct}^2z_{ct}+\tau_{ct}^2x_{ict}}{\sigma_{ct}^2+\tau_{ct}^2}}{\sqrt{\frac{\sigma_{ct}^2\tau_{ct}^2}{\sigma_{ct}^2+\tau_{ct}^2}}}\right]
\end{equation}

Let $P^*_{ct}=P(\theta_{ct}\leq\theta_{ct}^*|x_{ct}^*)$, so $P^*_{ct}$ is the agent's belief about the probability of a successful coup when the agent receives exactly the cutoff signal $x_{ct}^*$.  At the critical value of $x^*_{ct}$, the agent must be indifferent between engaging in a coup or not. Given payoffs in equation (\ref{payoffs}), this requires that:
\begin{equation} \label{exppayoff_eq}
P^*_{ct}b_{ct} - (1-P^*_{ct})k_{ct}=0\quad\Rightarrow\quad P^*_{ct} = \frac{1}{1+\frac{b_{ct}}{k_{ct}}}
\end{equation}

Substituting from (\ref{expprob}) into (\ref{exppayoff_eq}) and simplifying gives us the following expression defining the cutoff value of $\theta_{ct}^*$ given $x_{ct}^*$: 
\begin{equation} \label{crittheta}
\Phi\left[\frac{\frac{\sigma_{ct}}{\tau_{ct}}(\theta_{ct}^*-z_{ct}) + \frac{\tau_{ct}}{\sigma_{ct}}(\theta^*_{ct}-x_{ct}^*)}{\sqrt{\sigma_{ct}^2+\tau_{ct}^2}}\right]=\frac{1}{1+\frac{b_{ct}}{k_{ct}}}
\end{equation}
which, together with equation (\ref{attack_fail}) jointly determines $x_{ct}^*$ and $\theta_{ct}^*$. 

Equation (\ref{attack_fail}) and (\ref{crittheta}) also implicitly define equilibrium relationships between $z_{ct}$ and $b_{ct}$ and the condition for a coup to be successful ($\theta_{ct}\leq\theta_{ct}^*$). To take the model to the data, additional structure is required to distinguish empirically between the incentives for overthrow and the feasibility of successful collective action. As is often the case, identification requires a combination of normalization assumptions, exclusion restrictions, and parametric structure. We interpret the latent objects $z_{ct}$, $b_{ct}$, and in some cases $\sigma_{ct}$ as functions of observable country characteristics and impose normalizations on remaining model parameters ($k_{ct}$ and $\tau_{ct}$). This allows model terms to be estimated using observed coup attempts and success rates while preserving the structural interpretation of the underlying coordination problem. The resulting framework provides a way to distinguish whether observable determinants of coups operate primarily through changes in the expected benefits of overthrow or through changes in regime resilience and coordination feasibility.

\section{Empirical implementation}  \label{empact}

To develop an empirical implementation of the model, we adopt the standard panel framework used in the empirical literature on coups d'\'{e}tat. In the typical cross-country coup dataset (e.g., \citet{PowellThyne2011}) a panel of countries $c=\{1,2,\dots,C\}$ over years $t=\{1,2,\dots,T\}$ is observed. For each country-year, whether a coup attempt occurred during the year is observed, and, conditional on occurrence, whether the attempt succeeded is also known.\footnote{In actuality, more than one coup attempt might occur in a year; we describe how we address this situation in section \ref{sec:est}.} This aligns with our setup of the desirability-participation stage game in the previous section. Let $Y_{ct}\in\{0,1\}$ indicate coup occurrence and let $S_{ct}\in\{0,1\}$ indicate coup success. It is also possible that there may be additional information available about coups should they occur, and indeed, the typical coup dataset usually includes some additional information about the nature of the coup. Accordingly, we also consider a variation of the model in which coups are observed to differ in participation rate, and adapt some of the typical information about coups in the typical coup dataset to our model. The availability of such information has important implications for the identification and estimation of the $\sigma_{ct}$ term describing the variability in citizens' signals about regime quality. 

The model is characterized by the country-year specific fundamentals $z_{ct}$, $b_{ct}$, and $\sigma_{ct}$, and also $\kappa_{ct}$ and  $\tau_{ct}$. These denote, respectively, regime strength, the expected benefits of a successful coup to participants, the dispersion of agent's private information about regime strength, the expected costs associated with participation in a failed coup, and the dispersion of regime strength about its mean. Together, these parameters characterize both the desirability and feasibility of regime overthrow. Once they are known, other quantities of interest, such as the latent regime strength $\theta_{ct}$, participation thresholds, and equilibrium coup probabilities, can be recovered analytically or through simulation.

The task at hand is therefore to identify these latent components from observed coup attempts and outcomes. We model regime strength, coup benefits, and information quality as latent indices composed of observable covariates and unobserved components with specified distributions. Identification relies on the equilibrium restrictions imposed by the global-games framework, scale normalizations on the latent indices, and differential entry of observable covariates into the desirability and feasibility equations. Combined with the normal-normal information structure, these assumptions generate a tractable mapping from observed country characteristics to the probabilities of coup attempts and coup success.

In addition to Assumption \ref{nnmod} in section \ref{theorysec}, we adopt a linear index structure for mean coup benefits and mean regime strength based on a set of potential explanatory variables $X_{ct}$, as detailed in Assumptions \ref{meanb} and \ref{meanz}.

\begin{assumption}\label{meanb}
The latent desirability of regime overthrow is given by the index

\[
b_{ct}=X^b_{ct}\beta +\epsilon^b_{ct},
\]

where $X^b_{ct}\subseteq X_{ct}$ is a vector of observable covariates, $\beta$ is a conformable parameter vector, and

\[
\epsilon^b_{ct}\sim N(0,1).
\]

Thus $b_{ct}$ consists of an observed component, $X^b_{ct}\beta $, and an unobserved component, $\epsilon^b_{ct}$, implying

\[
b_{ct}\mid X^b_{ct}
\sim
N( X^b_{ct}\beta,1).
\]

The unit variance normalization fixes the scale of the latent benefit index.
\end{assumption}

\begin{assumption}\label{meanz}
The mean regime-strength index is given by
\[
z_{ct}= X^z_{ct}\zeta,
\]
where $X^z_{ct}\subseteq X_{ct}$ is a vector of observable covariates and $\zeta$ is a conformable parameter vector. Realized regime strength is given by
\[
\theta_{ct}=z_{ct}+u^\theta_{ct},
\qquad
u^\theta_{ct}\sim N(0,1).
\]
Equivalently,
\[
\theta_{ct}\mid X^z_{ct}
\sim
N(X^z_{ct}\zeta ,1).
\]
The unit variance normalization fixes the scale of the latent regime-strength index and implies $\tau_{ct}=1$.
\end{assumption}
The normal-normal information structure adopted in Assumption \ref{nnmod} can therefore be understood as a tractable reduced-form representation of the information environment facing potential coup participants. While such assumptions are commonly used in the global-games literature as a modeling device, their interpretation here is closer to that of latent-index models in empirical discrete-choice settings. Just as utility scales are identified only up to normalization, the level of uncertainty in the model is meaningful only relative to the scale of the underlying fundamentals. Consequently, $\sigma_{ct}$ is best viewed as a relative measure of information quality and the informativeness of the coordination environment rather than as the literal variance of a common signal received by all agents.

The independence and unit-variance assumptions imposed on the latent disturbances should be interpreted primarily as identifying normalizations rather than as substantive claims about political instability. The model explicitly permits common observable shocks to affect both the desirability and feasibility of regime overthrow through covariates that enter both latent indices. For example, economic conditions, institutional characteristics, and political factors may simultaneously influence expected coup benefits and regime strength. The independence assumption applies only to the residual components of these indices after conditioning on observables. Allowing unrestricted covariance or heteroskedasticity in the latent disturbances would introduce additional parameters whose identification from observed coup occurrence and coup outcomes alone is unclear. Indeed, a recurring theme of our approach is to make maximal use of information already available in existing coup datasets rather than to require the collection of entirely new forms of data. Where additional information is needed for identification, we seek to rely on simple extensions of variables already commonly coded in the empirical literature.
In any event, the unit variances therefore serve the conventional role of fixing the scale of the latent indices, analogous to standard latent-variable and discrete-choice models.

Assumption \ref{params} pins down relative benefits and costs by effectively defining $b_{ct}$ to be a cost-benefit ratio.

\begin{assumption} \label{params} $k_{ct}=1$.
\end{assumption}

Recall that our theory describes coups using a two-stage game where in the first stage, the desirability of a coup is determined by the expected benefits accruing to participants in a successful coup, denoted $b_{ct}$, and in the second stage, citizens decide whether to participate given these expected benefits. Under Assumption \ref{meanb}, a coup attempt occurs whenever the mean benefit from participation is nonnegative, so that coup occurrence may be represented as the event

\begin{equation}
Y_{ct}=\mathbf{1}(b_{ct}\geq 0).
\end{equation}

Combined with Assumption \ref{meanb}, this characterization implies a probit structure for coup attempts, as the probability that a coup attempt is observed in country $c$ at time $t$ is simply
\begin{equation}
\Pr\left( X^b_{ct}\beta+\epsilon^b_{ct}\geq 0\right)=
\Phi\left( X^b_{ct}\beta\right),
\end{equation}
while the probability that no coup attempt occurs is $\Phi\left(- X^b_{ct}\beta\right)$. Note that this assumption plays a role akin to an exclusion restriction, in that whether a coup attempt observed depends solely on $b_{ct}\geq 0$, and therefore plays a critical role in identifying the components of $b_{ct}$. 

Conditional on a coup attempt, the model proceeds to the participation stage, which determines the equilibrium threshold $\theta_{ct}^*$ governing coup success. To find this threshold, first solve equation (\ref{attack_fail}) for $x_{ct}^*$ to give:
\begin{equation} \label{invert}
x_{ct}^*=\theta_{ct}^*+\sigma_{ct}q(\theta_{ct}^*)
\end{equation}
where $q(Z)$ is the normal quantile function. Substitute (\ref{invert}) into (\ref{crittheta}), use the normalization $\tau_{ct}=1$ from Assumption \ref{meanz} and $k_{ct}=1$ from Assumption \ref{params} to get:

\begin{equation} \label{critthetarf}
\Phi\left[\frac{\sigma_{ct}(\theta_{ct}^*-z_{ct}) - q(\theta^*_{ct})}{\sqrt{\sigma_{ct}^2+1}}\right]=\frac{1}{1+b_{ct}}
\end{equation}

Given values for $b_{ct}$, $z_{ct}$, and $\sigma_{ct}$, the threshold $\theta_{ct}^*$ can be calculated from equation (\ref{critthetarf}). Under our parametric assumptions, the probability of coup success is then simply the probability that $\theta_{ct}\leq\theta_{ct}^*$, namely 

\begin{equation}
\Pr(\theta_{ct}\leq\theta_{ct}^*)=
\Phi\left(\theta_{ct}^*-z_{ct}\right).
\end{equation}

The threshold $\theta_{ct}^*$ is itself a function of the expected benefits from a coup, $b_{ct}$, so that the probability of success depends on both the desirability and feasibility of regime change. Because this dependence is closely related to the broader identification and estimation strategy, we defer further discussion until the next subsection. We now address the related issue of the identification of the dispersion of information parameter $\sigma_{ct}$.

In principle, information on coup participation or attack size can be used to identify this parameter. Equation (\ref{attack}) links the observed size of an attack to the precision of citizens' private information, so that participation data would provide direct information about $\sigma_{ct}$. In the absence of such information, however, $\sigma_{ct}$ cannot be separately recovered from observed coup attempts and outcomes. Put differently, the model may be viewed as a system of two latent equations: one governing participation and attack size, and a second governing coup success conditional on an attack. Identification of $\sigma_{ct}$ relies on observing both margins. 

This observation suggests two alternative empirical specifications, depending on whether only coup success is observed, or both success and attack strength are observed. Our first specification follows a scenario in which only success or failure of a coup is observed conditional on occurrence, and follows the common approach in the theoretical literature in which the limiting case $\sigma_{ct}\rightarrow 0$ \citep{morris2003} is used in the analysis. As shown by \citet{morris2003}, equilibrium uniqueness is preserved in this limit while the characterization of equilibrium behavior becomes considerably simpler. Econometrically, the resulting model depends only on the latent benefit and regime-strength components, yielding a parsimonious framework for studying coup attempts and success.

The second specification follows \citet{angeletos2007a} and \citet{edmond13} in treating information quality as an economically meaningful component of the model. In this case, $\sigma_{ct}$ is estimated directly and measures the precision with which citizens observe regime strength. The parameter therefore admits a natural interpretation as the quality of information available to potential coup participants and may itself be of substantive interest. Accordingly, we consider both specifications in the empirical analysis: a ``simple'' model imposing $\sigma_{ct}\rightarrow 0$ and a more general model (the ``full'' model) in which $\sigma_{ct}$ is estimated alongside the remaining latent components.

In circumstances in which we wish to estimate $\sigma_{ct}$, we must deploy an additional assumption to ensure model coherence:

\begin{assumption} [Uniqueness of $\theta_{ct}^*$] \label{sigass} $\sigma_{ct} \leq \sqrt{2\pi}$. \label{unique}
\end{assumption}

Assumption \ref{sigass} is standard in the global games literature; see, for example, the discussion in \citet[Section 3.1]{morris2003}. If Assumption \ref{sigass} is violated, equation (\ref{crittheta}) may admit multiple solutions: as \citet[p. 82]{morris2003} put it, the model hits ``the multiplicity zone'' and in context means that no longer is the model selecting unique equilibria.\footnote{To see uniqueness, rewrite (\ref{crittheta}) as
\[
\sigma_{ct}\theta^*_{ct}-q(\theta^*_{ct})
=
\sigma_{ct}z_{ct}
+
\sqrt{1+\sigma_{ct}^2}
q\!\left(\frac{1}{1+b_{ct}}\right),
\]
where the right-hand side is constant. Let $x=q(\theta^*_{ct})$, so that $\theta^*_{ct}=\Phi(x)$. The left-hand side becomes $\sigma_{ct}\Phi(x)-x$, with derivative $\sigma_{ct}\phi(x)-1$. Since $\phi(x)\leq 1/\sqrt{2\pi}$, the derivative is negative for all $x$ whenever $\sigma_{ct}<\sqrt{2\pi}$, implying a unique solution.} While the assumption is needed for parametric identification of $\sigma_{ct}$, this leaves the issue of nonparametric identification of the variance parameter with what is effectively a model with discrete outcomes. In the next section, we first develop a simple version of the model in which $\sigma_{ct}\rightarrow 0$. In section \ref{fullmodel}, we discuss our approach to the identification of $\sigma_{ct}$ in greater detail. 

\subsection{Simple model} \label{simpmodel}

In this version of the model, $\sigma_{ct}$ is assumed to approach zero. As is shown in \citet{morris2003}, this results in a simple structure for both the critical value of the state $\theta^*_{ct}$ and the structure of the model. The limiting model should not be viewed simply as a theoretical approximation. Rather, it corresponds to the empirically relevant case in which only coup occurrence and coup success are observed. In the absence of information on attack size, the information-dispersion parameter $\sigma_{ct}$ is not separately identified, making the limiting model a natural benchmark specification. At the same time, the limiting case $\sigma_{ct}\rightarrow 0$ plays a central role in the theoretical global games literature, where it serves as a canonical benchmark to characterize equilibrium behavior while preserving the equilibrium-selection properties of incomplete information \citep{morris2003,angeletos2007a}.

Proposition \ref{simplecoup} characterizes the result as manifest in the current circumstances:

\begin{proposition} \label{simplecoup} Under Assumptions 1, 2, and 3, and allowing $\sigma_{ct} \rightarrow 0$, when $b_{ct}>0$:
\begin{equation}
\theta_{ct}^* =\frac{b_{ct}}{1+ b_{ct}}.   
\end{equation}
Then  
\begin{itemize}
    \item  When $\theta_{ct}\leq\theta_{ct}^*$, $A_{ct}=1$ and the coup is successful, and
    \item When $\theta_{ct}>\theta_{ct}^*$, $A_{ct}=0$ and the coup fails. 
\end{itemize} 

\end{proposition} 
\begin{proof}[Proof:]
Rewrite condition (\ref{crittheta}) using $\Phi[-z]=1-\Phi[z]$, let $\sigma_{ct}\rightarrow 0$, and solve for $\theta_{ct}^*$. 
\end{proof}

Proposition \ref{simplecoup} describes the empirical content of the basic ideas in \citet{morris2003}. As parametric uncertainty vanishes, the model leaves strategic uncertainty resolution intact, which results in all agents coordinating on attacking the regime (conditional on $b_{ct}>0$) when $\theta_{ct}\leq\theta_{ct}^*$, and no agents attacking otherwise. While straightforward, this version of the model poses a possible conceptual difficulty: coup participation is either full or zero.\footnote{On the other hand, it squares quite well with the logical result that if a coup is known to be unsuccessful and there is no uncertainty about regime strength, no one would participate in it.} Another way to view Proposition \ref{simplecoup} is as a statement as to what the model can offer information on if nothing is known about the strength of an attack. If attack strength $A_{ct}$ is unobserved, it is natural to suppose that participation is all-or-nothing, and there is no basis for estimating the parameter $\sigma_{ct}$.

The probability that a coup is observed and it is successful depends upon the exact value of $b_{ct}$, as this determines $\theta_{ct}^*(b_{ct})$, which in turn determines the probability of a successful coup (see Proposition \ref{simplecoup}). That is, conditional on $b_{ct} \geq 0$,  we have:

\begin{equation} \label{nocoupc}
    P(\textrm{successful coup}|\textrm{coup attempt}) = P(\theta_{ct}\leq\theta_{ct}^*(b_{ct})|b_{ct}>0)
\end{equation}

Where:

\begin{equation} \label{nocoupd}
    P(\theta_{ct}\leq\theta_{ct}^*(b_{ct})|b_{ct}>0) = \int_{0}^\infty\Phi\left[\theta_{ct}^*(b')-z_{ct}\right]\phi(b', X_{ct}^b\beta, 1;0)db'
\end{equation}

In expression (\ref{nocoupd}), $\phi(b, X_{ct}^b\beta, 1;0)$ denotes the normal density with mean $ X_{ct}^b\beta$ and standard deviation 1, left truncated at zero. In (\ref{nocoupd}), we have used the assumption that $\tau_{ct}$ is assumed to be $N(z_{ct},1)$ (Assumption \ref{meanz}). 

The right-hand side of (\ref{nocoupd}) can be simulated using $K$ draws from the standard normal distribution with mean $ X_{ct}^b\beta$, variance of one, left-truncated at zero, while also applying our expression for the prior mean $z_{ct}=X_{ct}^z\zeta$ as follows:\footnote{As a practical matter, the simulated integrals (\ref{nocoupc}) and (\ref{nocoupd}) can use fixed draws from a uniform distribution and the desired value of $b$ can be computed using $b_{k} = \Phi^{-1}( \Phi(- X^b_{ct}\beta) + \Phi(X_{ct}^b)u_{kct})+ X^b_{ct}\beta.$} 

\begin{equation} \label{simplesuccoup}
    P(\theta_{ct}\leq\theta_{ct}^*(b_{ct})|b_{ct}>0) \approx K^{-1}\sum_{k=1}^K\Phi\left[\theta_{ct}^*(b_k)- X_{ct}^z\zeta\right]
\end{equation}

A similar approach yields the simulated probability of observing an unsuccessful coup, conditional on a coup attempt:

\begin{equation} \label{simplesuccoup2}
    P(\textrm{failed coup}|b_{ct}) \approx K^{-1}\sum_{k=1}^K\Phi\left[  X_{ct}^z\zeta-\theta_{ct}^*(b_k)\right]
\end{equation}

We can then write the log-likelihood contribution of an observation as:

\begin{align}
LL_{ct}
&= (1-Y_{ct})\ln \Phi\left[- X_{ct}^b\beta\right]
+ Y_{ct}\ln \Phi\left[ X_{ct}^b\beta\right] \nonumber \\
&\quad + Y_{ct}\Bigg(
S_{ct}\ln\left[
K^{-1}\sum_{k=1}^K
\Phi\left[\theta_{ct}^*(b_k)- X_{ct}^z\zeta\right]
\right] \nonumber \\
&\qquad\qquad
+ (1-S_{ct})\ln\left[
K^{-1}\sum_{k=1}^K
\Phi\left[ X_{ct}^z\zeta-\theta_{ct}^*(b_k)\right]
\right]
\Bigg).
\label{simplell}
\end{align}

and the model can be estimated by choosing $\beta$ and $\zeta$ to maximize 
\begin{equation*}
    LL = \sum_{c=1}^C\sum_{t=1}^T LL_{ct}.  
\end{equation*}

\subsection{Full model} \label{fullmodel}

To estimate a version of the model which includes $\sigma_{ct}$, a data concern and a computational concern must be addressed. The identification strategy for $\sigma_{ct}$ relies on observing information about the scale of coup participation. Ideally, one would observe the fraction of eligible participants actively involved in an attempted overthrow, permitting direct measurement of attack size through equation (19). Such information is rarely available in cross-national coup datasets. Instead, continuing a theme of the paper in using information already available in most coup data sets, we exploit qualitative information on coup strength already recorded in the Powell-Thyne data \citep{PowellThyne2011}. The dataset distinguishes coups that are broadly supported and involve substantial participation from those that are narrowly organized conspiratorial actions. We map these classifications into the theoretical categories of the model by treating coups identified as popular or broadly supported as strong attacks and coups identified as conspiratorial as weak attacks. Intermediate cases are classified as neither strong nor weak.

This mapping should not be interpreted literally as a direct measure of participation rates. Rather, it provides an ordinal measure of attack size sufficient to distinguish broad-based mobilization from narrowly organized attempts. To operationalize the model, we introduce thresholds $\eta_w$ and $\eta_s$ and interpret weak attacks as those involving participation rates below $\eta_w$ and strong attacks as those involving participation rates above $\eta_s$. In the empirical implementation, we set $\eta_w=0.1$ and $\eta_s=0.9$, so that weak coups correspond to attempts involving only a small fraction of potential participants while strong coups correspond to attempts attracting near-universal participation within the relevant coalition. 

Conditional on $\theta_{ct}$ and an equilibrium cutoff signal $x_{ct}^*$, attack size is defined as:

\begin{equation} \label{attacksize1}
A_{ct} = \Phi\left[\frac{x_{ct}^*-\theta_{ct}}{\sigma_{ct}}\right]
\end{equation}

Equation (\ref{invert}) describes the equilibrium relationship between $x_{ct}^*$ and $\theta_{ct}^*$: 
\begin{equation} \label{quantile_fn2}
x_{ct}^*=\theta_{ct}^* + \sigma_{ct} q (\theta_{ct}^*)
\end{equation}

Combining (\ref{attacksize1}) and (\ref{quantile_fn2}) gives a reduced-form relationship between $\theta_{ct}$, $\theta^*_{ct}$ and attack size:

\begin{equation} \label{attacksize_rf}
A_{ct} = \Phi\left[\frac{\theta_{ct}^* + \sigma_{ct} q (\theta_{ct}^*)-\theta_{ct}}{\sigma_{ct}}\right]
\end{equation}

Expression (\ref{attacksize_rf}) recasts statements about attack size in terms of $\theta_{ct}$. An attack is classified as strong if $A_{ct}\geq\eta_s$ and weak if $A_{ct}\leq\eta_w$. Using (\ref{attacksize_rf}), a strong attack is observed if:
\begin{equation}
\Phi\left[\frac{\theta_{ct}^* + \sigma_{ct} q (\theta_{ct}^*)-\theta_{ct}}{\sigma_{ct}}\right]  \geq \eta_s\quad\rightarrow\quad \theta_{ct} \leq \theta_{ct}^* + \sigma_{ct} \left(q(\theta^*_{ct})-q(\eta_s)\right)
\end{equation}
Similarly, we can develop the criteria for a weak attack being observed:
\begin{equation}
\Phi\left[\frac{\theta_{ct}^* + \sigma_{ct} q (\theta^*_{ct})-\theta_{ct}}{\sigma_{ct}}\right]  \leq \eta_w\quad\rightarrow\quad \theta_{ct} \geq \theta_{ct}^* + \sigma_{ct} \left(q(\theta_{ct}^*)-q(\eta_w)\right)
\end{equation}

Depending on the relationship between $\theta_{ct}^*$, $\eta_w$ and $\eta_s$, we can observe different combinations of attack strength and attack success or failure. A visual representation of these relationships is provided in Figure \ref{parmranges}, which illustrates how information on coup size helps identify $\sigma_{ct}$. The horizontal axis measures the critical threshold $\theta_{ct}^*$ required to overthrow the regime, while the vertical axis measures the realized value of $\theta_{ct}$. Below the 45-degree line, the regime is overthrown because $\theta_{ct}<\theta_{ct}^*$; above the line, the coup fails.

The introduction of weak and strong coup categories partitions this space further. For example, to the left of the locus defined by $\theta_{ct}^*+\sigma_{ct}\left[q(\theta_{ct}^*)-q(\eta_w)\right]$, participation falls below the weak-coup threshold, so that $A_{ct}<\eta_w$. This creates regions in which both successful and failed coups may be classified as weak. Analogous regions arise for strong coups when participation exceeds the threshold $\eta_s$. In fact, one might view the six regions of the figure determined by the loci on the figure as bins collecting the number of successful and failed coups of each type. By determining the slope of the loci on the figure, $\sigma_{ct}$ determines the relative size of each of these bins. 

The figure also illustrates the natural tendency for weak coups to fail and strong coups to succeed. Identification of $\sigma_{ct}$ comes from deviations from these patterns. In particular, observing weak coups that nevertheless succeed, or strong coups that nevertheless fail, provides information about the dispersion of private information. Such outcomes reflect coordination errors in the sense that participation is unexpectedly low in coups that ultimately succeed, or unexpectedly high in coups that ultimately fail.

\begin{figure}[htb!] 
\centering
\includegraphics[width=\textwidth]{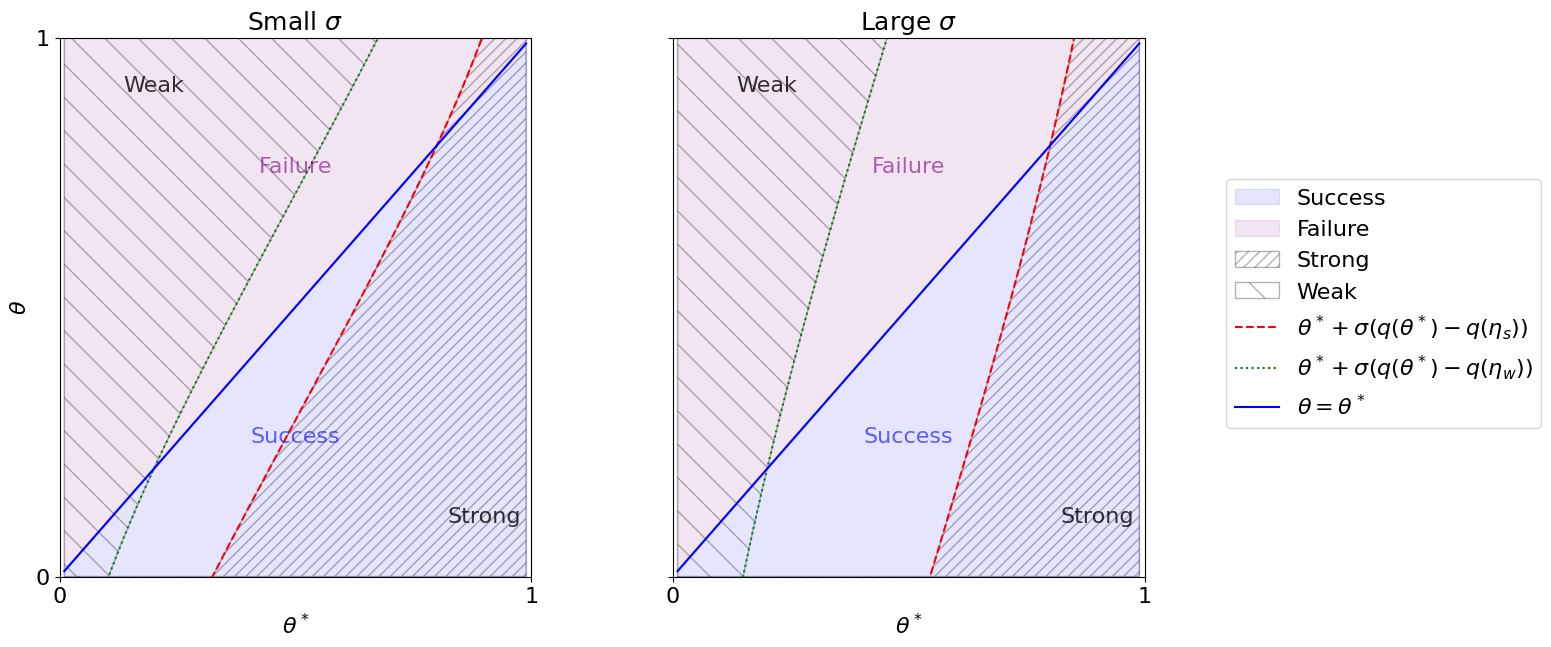}
\caption{Parameters showing possible configurations of $\theta^*$, and the participation parameters $\eta_w$ and $\eta_s$. Depending on where the actual value of $\theta$ falls and where in the range $\theta^*$ is, strong or weak coups that fail or succeed, or coups that are neither strong nor weak. Larger values of $\sigma$ alter the frequency with which different strengths are observed. Note that as $\sigma$ shrinks to zero, the dashed lines collapse to the blue line, distinguishing success from failure.}
\label{parmranges}
\end{figure}

The figure allows a direct means of calculating the likelihood of observing both the strength of the coup and its success or failure. For example, a strong, successful coup requires the following two criteria to hold: 

\begin{eqnarray} \label{strongsuc}
\theta_{ct} &\leq & \theta_{ct}^*, \nonumber \quad \textrm{and}\\
\theta_{ct} &\leq & \theta_{ct}^*+\sigma_{ct}[q(\theta_{ct}^*)-q(\eta_s)]
\end{eqnarray}

This requires that: 

\begin{equation} \label{strongsuc2}
\theta_{ct} \leq \theta_{ct}^* + \sigma_{ct}\min \left[0, q(\theta^*_{ct})-q(\eta_s)\right]
\end{equation}

An intuitive explanation of equation (\ref{strongsuc2}) is that when $\theta^*_{ct}$ is small, on the left-hand sides of the panels in Figure \ref{parmranges}, coups occur and are successful, yet aren't necessarily large. On the right side of each figure, however, $\theta_{ct}^*$ is close to one, so almost every successful coup will have to be a large one in this parameter range. 

Guided by Figure \ref{parmranges}, we can write the full likelihood for combinations of success, failure, and coup size, so that the six separate regions of Figure \ref{parmranges} are exhausted. To do so, we also introduce two dichotomous variables, $S_{ct} = 1$ if a coup is strong (and zero otherwise), and $W_{ct}=1$ if a coup is weak (zero otherwise). It also helps to introduce the simplifying notation:
\begin{eqnarray*}
q_{ct,s}^{+} = \max \left[0, q(\theta^*_{ct})-q(\eta_s)\right], & &q_{ct,s}^{-} = \min \left[0, q(\theta^*_{ct})-q(\eta_s)\right] \nonumber \\  
q_{ct,w}^{+} = \max \left[0, q(\theta^*_{ct})-q(\eta_w)\right], & &q_{ct,w}^{-} = \min \left[0, q(\theta^*_{ct})-q(\eta_w)\right]
\end{eqnarray*}
We then have a list of probabilities that we observe various combinations of regime strength and coup success, as follows:
\begin{align} \label{full_probs}
p_{ct}^{SS}
&\equiv
P(\textrm{strong}\cap\textrm{success}\mid b>0)
=
\Phi\left[
\theta_{ct}^*
+\sigma_{ct}q_{ct,s}^-
- X_{ct}^z\zeta
\right], \nonumber \\[0.5em]
p_{ct}^{SF}
&\equiv
P(\textrm{strong}\cap\textrm{failure}\mid b>0)
=
\Phi\left[
\theta_{ct}^*
+\sigma_{ct}q_{ct,s}^+
-X_{ct}^z\zeta
\right]
-
\Phi\left[
\theta_{ct}^*
-X_{ct}^z\zeta
\right], \nonumber \\[0.5em]
p_{ct}^{WS}
&\equiv
P(\textrm{weak}\cap\textrm{success}\mid b>0)
=
\Phi\left[
\theta_{ct}^*
-X_{ct}^z\zeta
\right]
-
\Phi\left[
\theta_{ct}^*
+\sigma_{ct}q_{ct,w}^-
-X_{ct}^z\zeta
\right], \nonumber \\[0.5em]
p_{ct}^{WF}
&\equiv
P(\textrm{weak}\cap\textrm{failure}\mid b>0)
=
\Phi\left[
X_{ct}^z\zeta
-
\left(
\theta_{ct}^*
+\sigma_{ct}q_{ct,w}^+
\right)
\right], \nonumber \\[0.5em]
p_{ct}^{NS}
&\equiv
P(\textrm{neither}\cap\textrm{success}\mid b>0)
=
\Phi\left[
\theta_{ct}^*
+\sigma_{ct}q_{ct,w}^-
-X_{ct}^z\zeta
\right]
-
\Phi\left[
\theta_{ct}^*
+\sigma_{ct}q_{ct,s}^-
-X_{ct}^z\zeta
\right], \nonumber \\[0.5em]
p_{ct}^{NF}
&\equiv
P(\textrm{neither}\cap\textrm{failure}\mid b>0)
=
\Phi\left[
\theta_{ct}^*
+\sigma_{ct}q_{ct,w}^+
-X_{ct}^z\zeta
\right]
-
\Phi\left[
\theta_{ct}^*
+\sigma_{ct}q_{ct,s}^+
-X_{ct}^z\zeta
\right].
\end{align}
Let
\begin{equation*}
\mathcal J=\{SS,SF,WS,WF,NS,NF\},
\end{equation*}
where $SS$ denotes strong-success, $SF$ strong-failure, $WS$ weak-success, $NS$ denotes a neither strong-nor-weak coup that was nonetheless successful, and so on. Let $J_{ct}\in\mathcal J$ denote the observed coup category, conditional on $Y_{ct}=1$. Because the equilibrium threshold $\theta^*_{ct}$ depends on the realized value of $b_{ct}$, the probabilities in (\ref{full_probs}) must be integrated over the distribution of $b_{ct}$ conditional on a coup attempt. We approximate this integral using the same truncated-normal simulation draws used in the simple model:
\[
\bar p_{ct}^{j}
=
K^{-1}\sum_{k=1}^{K}
p_{ct}^{j}(b_k, X_{ct}^{z}\zeta,\sigma_{ct}),
\qquad
b_k\sim b_{ct}\mid b_{ct}>0,X_{ct}^{b}.
\]

The likelihood contribution for country $c$ in period $t$ is then
\begin{equation} \label{full_ll}
LL_{ct}
=
(1-Y_{ct})\ln \Phi\!\left(- X_{ct}^{b}\beta\right)
+
Y_{ct}\ln \Phi\!\left( X_{ct}^{b}\beta\right)
+
Y_{ct}
\sum_{j\in\mathcal J}
\mathbf{1}\{J_{ct}=j\}\ln \bar p_{ct}^{j}.
\end{equation}

and the model can be estimated by choosing $\beta$ and $\zeta$ to maximize 
\begin{equation*}
    LL = \sum_{c=1}^C\sum_{t=1}^T LL_{ct}.  
\end{equation*}

Estimation also requires repeated computation of $\theta^*_{ct}$. Equation (\ref{crittheta}) can be rewritten in a form leading to a contraction mapping that exploits the smoothness, slope and boundedness properties of the cumulative normal distribution:
\begin{equation} \label{cmap}
\theta_{ct,m+1}^* = \Phi\left[\frac{\sigma_{ct}}{\tau_{ct}^2}\left( \theta_{ct,m}^*-z_{ct}\right)-q\left(\frac{1}{1+b_{ct}}\right)\frac{\sqrt{\sigma_{ct}^2+\tau_{ct}^2}}{\tau_{ct}}\right]
\end{equation}

Under Assumption \ref{sigass} iteration of (\ref{cmap}) converges to $\theta_{ct}^*$ via the contraction mapping theorem.

\section{Estimation} \label{sec:est}

Our data are drawn from the Cline Center Coup d'\'{E}tat Project dataset maintained by \citet{cline}.\footnote{The data are publicly available at \url{https://doi.org/10.13012/B2IDB-9651987_V6}.} The dataset covers approximately 130 countries from 1950 through 2021, yielding roughly 11,000 country-year observations. According to the project codebook, a coup d'\'{e}tat is defined as an ``organized effort to effect sudden and irregular removal of the incumbent executive authority of a national government, or to displace authority at the highest levels of one or more branches of government'' \citep[p.~3]{peyton2026cline}.

The data record the country and year in which coup attempts occur, as well as whether those attempts are successful. Accordingly, the data may be organized as a country-year panel in which each observation indicates whether a coup attempt occurred and, conditional on an attempt, whether it succeeded. Figure \ref{fig} summarizes the evolution of coup attempts and successes over time, while Figure \ref{fig2} summarizes the geographic distribution of coups over time.\footnote{These graphs may differ slightly from those produced by the online graphing and mapping tool provided by \citet{ClineCenter} at \url{https://clinecenter.illinois.edu}. The difference arises because our figures aggregate coup-years rather than individual coup events, consistent with the structure of our panel dataset.}

\begin{figure}[!htbp]
\centering
\includegraphics[width=.9\textwidth]{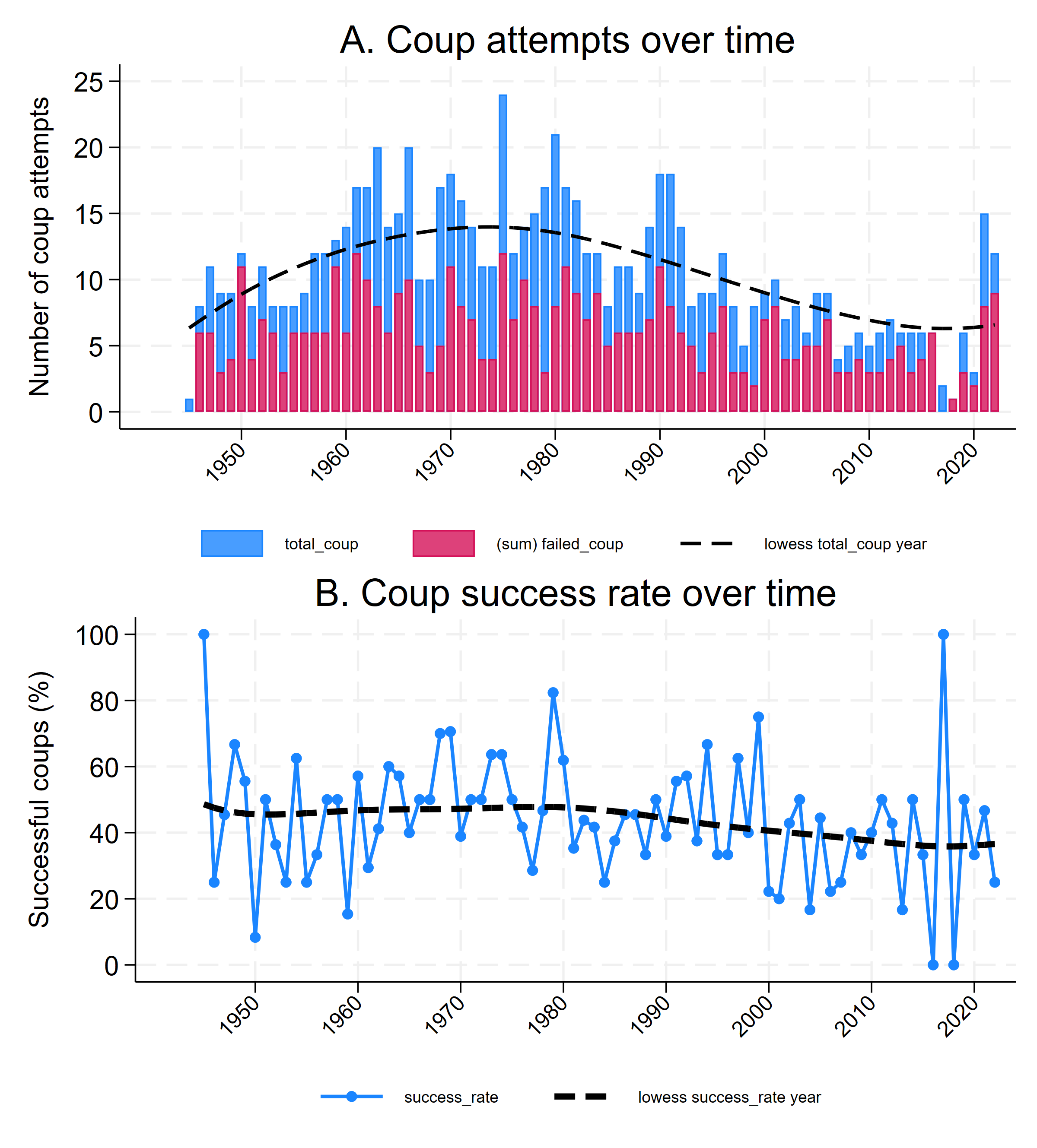}
\caption{Coup attempts and success over time}
\label{fig}
\end{figure}

On average, there were 10.76 coup-years, corresponding to coup activity in roughly 8\% of countries in a given year.\footnote{The Cline center database in fact records whether there were multiple coups in a given year; we have collapsed this information to a simple binary occurrence and success variables for each country-year pair.} Figure \ref{fig} contains a quadratic trend line, which suggests that the frequency of coups has remained relatively stable over the sample period with a modest downward trend in recent decades. The lower panel of the figure suggests that the success rate of coups appears to have remained comparatively stable over time. Figure \ref{fig2} displays the geographic distribution of coups, where darker colors indicate more coup activity. The figure suggests that coup activity is more prevalent in less-developed regions of the world. While purely descriptive, this pattern is consistent with the broader empirical literature linking coups to weaker macroeconomic conditions. 

Indeed, the empirical literature on coups d'\'{e}tat is vast and includes hundreds of proposed determinants of coup incidence and success. Rather than attempting to encompass the full universe of variables considered in this literature, our objective is to construct parsimonious specifications that capture the broad classes of factors most consistently identified as important. In this sense, our approach is closer in spirit to the synthesis provided by studies such as \citet{gassebner2016expect} and \citet{powell2012} than to an exhaustive search over possible correlates of coup activity.

\begin{figure}[!htbp]
\centering
\includegraphics[width=.9\textwidth]{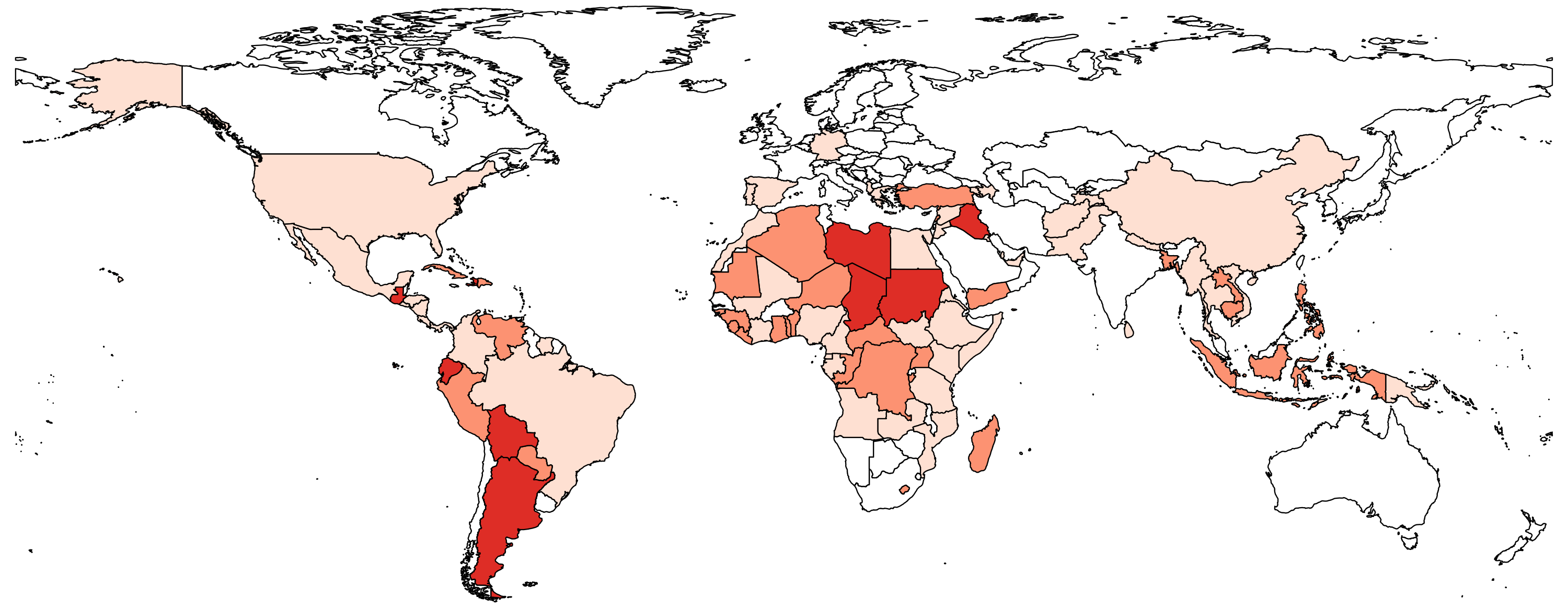}
\caption{Coup attempts and success over time}
\label{fig2}
\end{figure}

We therefore supplement the coup data with a set of country-level covariates intended to represent the principal categories of explanatory variables emphasized in the literature. Economic performance is captured using the growth rate of real GDP per capita and the log of real GDP per capita, both obtained from the Penn World Tables \citep{feenstra2015}. Institutional persistence is measured using regime durability, defined as the number of years since the most recent regime transition and drawn from the Polity5 dataset \citep{marshall2020}. To capture the ability of regimes to deter or withstand collective action, we include a measure of coup-proofing based on institutional arrangements designed to reduce coup risk, following \citet{quinlivan99}. Finally, we incorporate an index of information reliability, adapted from the informational framework developed in \citet{edmond13}, intended to capture the precision of information available to agents engaged in strategic coordination. Summary statistics for the variables employed in the analysis are reported in Table \ref{tab:summarystats}.

\begin{table}[!htbp]
\centering
\scriptsize
\caption{Summary statistics by coup d'\'etat occurrence and outcome}\vspace{.3cm}
\label{tab:summarystats}
{
\def\sym#1{\ifmmode^{#1}\else\(^{#1}\)\fi}
\begin{tabular}{l*{4}{ccc}}
\toprule
                    &\multicolumn{3}{c}{\shortstack{Full sample\\ (N=10764)}}&\multicolumn{2}{c}{\shortstack{No coup\\ (N=9925)}}&\multicolumn{2}{c}{\shortstack{Failed coup\\ (N=463)}}&\multicolumn{2}{c}{\shortstack{Successful coup\\ (N=376)}}\\
                    &        Obs.&        Mean&          SD&        Mean&          SD&        Mean&          SD&        Mean&          SD\\
\midrule
Real GDP growth     &        7238&       0.017&     (0.067)&       0.019&     (0.063)&       0.002&     (0.087)&      -0.021&     (0.100)\\
Log GDP             &        7364&       8.539&     (1.178)&       8.589&     (1.185)&       8.040&     (0.963)&       8.028&     (0.964)\\
Population growth   &        7238&       0.020&     (0.016)&       0.020&     (0.016)&       0.024&     (0.014)&       0.020&     (0.024)\\
Durability          &        7838&      16.066&    (19.784)&      16.869&    (19.875)&      11.471&    (19.413)&       5.653&    (13.855)\\
Polity score        &        7874&      -0.776&     (7.002)&      -0.516&     (7.101)&      -2.390&     (6.018)&      -4.023&     (4.668)\\
Coup proofing       &        5987&       1.749&     (0.655)&       1.754&     (0.655)&       1.673&     (0.656)&       1.719&     (0.654)\\
Information reliability&        9689&       0.127&     (0.847)&       0.117&     (0.855)&       0.209&     (0.808)&       0.269&     (0.666)\\
\bottomrule
\end{tabular}
}

\end{table}

Several patterns emerge from the summary statistics reported in Table \ref{tab:summarystats}. Unsurprisingly, country-year observations experiencing coup attempts differ systematically from observations in which no coup occurs, which may account in large part for differences in the rate at which coups occur across regions and time. Coup attempts are associated with lower Polity scores on average, indicating that they are more common in less democratic political environments. They are also associated with higher values of the information reliability index, consistent with the possibility that strategic coordination becomes easier when agents possess more precise information. Among country-year observations experiencing coup attempts, successful coups tend to be associated with lower rates of economic growth and less favorable measures of governmental performance than failed coups. Successful coups also occur, on average, in less durable regimes, suggesting that younger political institutions may be more vulnerable to successful challenges.

Table \ref{tab:summarystats} also highlights an important and familiar practical consideration for the empirical analysis: The inclusion of additional covariates reduces the number of available observations because of missing data. While this issue is common in cross-country panel studies, it is particularly consequential here because coup attempts are relatively rare events. This consideration further motivates the use of parsimonious specifications.

\subsection{Simple Models}

Table \ref{tab:table1} reports estimates of several restricted specifications corresponding to the limiting case in which $\sigma_{ct}\rightarrow 0$. In this setting, regime strength is effectively observed by all agents, eliminating informational uncertainty and reducing the model to one in which strategic uncertainty alone determines participation decisions. Table \ref{tab:table1} leads off with an intercept-only specification, which provides a calibration benchmark for the latent-index model. In the benefit equation, the estimated intercept satisfies $\Phi(\hat\beta_0)\approx 0.068$, so that the implied probability of a coup attempt matches the unconditional frequency observed in the data. Likewise, the intercepts imply an unconditional probability of coup success, $\int_0^\infty \Phi\left(\tfrac{b}{1+b}-\hat z_0\right)f(b\mid b>0)\,db \approx 0.446$, matching the frequency in the data.

\begin{table}[!htbp]
\centering
\footnotesize
\caption{Estimated models of mean benefits and regime strength}
\vspace{.3cm}
\label{tab:table1}
{
\def\sym#1{\ifmmode^{#1}\else\(^{#1}\)\fi}
\begin{tabular}{l*{5}{c}}
\toprule
            &\multicolumn{1}{c}{(1)}&\multicolumn{1}{c}{(2)}&\multicolumn{1}{c}{(3)}&\multicolumn{1}{c}{(4)}&\multicolumn{1}{c}{(5)}\\
            &\multicolumn{1}{c}{Constants}&\multicolumn{1}{c}{Time Trend}&\multicolumn{1}{c}{+ Proofing Excl.}&\multicolumn{1}{c}{+ Protest Excl.}&\multicolumn{1}{c}{Parsimonious}\\
\midrule
\textbf{Coup benefits $b_{ct}$}&                     &                     &                     &                     &                     \\
Constant    &      -1.419\sym{***}&       7.017\sym{***}&       27.62\sym{***}&       21.40\sym{**} &       27.29\sym{***}\\
            &    (0.0177)         &     (1.598)         &     (4.621)         &     (9.667)         &     (4.607)         \\
Year        &                     &    -0.00426\sym{***}&     -0.0137\sym{***}&     -0.0105\sym{**} &     -0.0136\sym{***}\\
            &                     &  (0.000807)         &   (0.00231)         &   (0.00484)         &   (0.00231)         \\
Real GDP per capita growth&                     &                     &      -1.420\sym{***}&      -1.120\sym{**} &      -1.427\sym{***}\\
            &                     &                     &     (0.355)         &     (0.498)         &     (0.356)         \\
Log GDP per capita&                     &                     &      -0.176\sym{***}&      -0.240\sym{***}&      -0.175\sym{***}\\
            &                     &                     &    (0.0257)         &    (0.0382)         &    (0.0254)         \\
Regime durability&                     &                     &    -0.00652\sym{***}&    -0.00737\sym{**} &    -0.00646\sym{***}\\
            &                     &                     &   (0.00185)         &   (0.00302)         &   (0.00185)         \\
Polity Index&                     &                     &     -0.0130\sym{***}&     -0.0216\sym{***}&     -0.0119\sym{**} \\
            &                     &                     &   (0.00488)         &   (0.00693)         &   (0.00472)         \\
Population growth rate&                     &                     &      -1.951         &       0.101         &                     \\
            &                     &                     &     (2.080)         &     (2.854)         &                     \\
Log protests&                     &                     &                     &       0.176\sym{***}&                     \\
            &                     &                     &                     &    (0.0555)         &                     \\
\midrule
\textbf{Regime strength $z_{ct}$}&                     &                     &                     &                     &                     \\
Constant    &       0.400\sym{***}&      -4.787         &      -32.57\sym{***}&      -30.17         &      -35.08\sym{***}\\
            &    (0.0440)         &     (4.424)         &     (12.19)         &     (25.94)         &     (11.12)         \\
Year        &                     &     0.00262         &      0.0166\sym{***}&      0.0157         &      0.0178\sym{***}\\
            &                     &   (0.00224)         &   (0.00608)         &    (0.0130)         &   (0.00560)         \\
Real GDP per capita growth&                     &                     &       0.101         &       1.231         &                     \\
            &                     &                     &     (0.931)         &     (1.430)         &                     \\
Log GDP per capita&                     &                     &    -0.00708         &     -0.0303         &                     \\
            &                     &                     &    (0.0679)         &     (0.102)         &                     \\
Regime durability&                     &                     &     0.00272         &    0.000415         &                     \\
            &                     &                     &   (0.00401)         &   (0.00854)         &                     \\
Polity Index&                     &                     &     0.00710         &    0.000150         &                     \\
            &                     &                     &    (0.0131)         &    (0.0198)         &                     \\
Population growth rate&                     &                     &       18.94\sym{***}&       15.32\sym{*}  &       18.41\sym{***}\\
            &                     &                     &     (6.030)         &     (8.180)         &     (5.871)         \\
Coup-proofing measures&                     &                     &      -0.183\sym{*}  &      -0.434\sym{***}&      -0.178         \\
            &                     &                     &     (0.111)         &     (0.155)         &     (0.110)         \\
\midrule
\(N\)       &       10764         &       10764         &        5076         &        3243         &        5076         \\
\bottomrule
\multicolumn{6}{l}{\footnotesize Notes: Estimation follows equation (\ref{simplell}) with $K=100$.}\\
\multicolumn{6}{l}{\footnotesize All explanatory variables are lagged one year.}\\
\multicolumn{6}{l}{\footnotesize *, **, and *** denote significance at the 10\%, 5\%, and 1\% levels.}\\
\end{tabular}
}

\end{table}

To identify the regime-strength equation, we employ an exclusion restriction based on coup-proofing measures. Specifically, coup-proofing is excluded from the coup-benefit equation, $b_{ct}$, and included only in the regime-strength equation, $z_{ct}$. This restriction is motivated by the view that coup-proofing arrangements affect the ability of a regime to withstand a coup attempt without directly altering the benefits that participants expect from a successful coup. All explanatory variables are lagged one year and therefore predate the realization of coup outcomes in period $t$. This timing convention reflects the possibility that economic conditions, political institutions, and coup-proofing measures may themselves respond to political instability (e.g., \citet{barro1991}, \citet{alesina1996}). The objective of the paper is not to model these broader dynamic interactions, but rather to examine how previously observed country characteristics map into the desirability and feasibility of regime overthrow within the global-games framework. Accordingly, the covariates should be interpreted as predetermined state variables rather than as causally identified determinants of coup activity.

Several features of the estimates are noteworthy. First, the explanatory variables exhibit much stronger and more precisely estimated effects in the coup-benefit equation, $b_{ct}$, than in the regime-strength equation, $z_{ct}$. This difference is likely attributable to the fact that the benefit equation is informed by all country-year observations, whereas the regime-strength equation is identified only from observations in which a coup attempt occurs.

The estimated effects in the benefit equation are generally intuitive. Higher levels of income and faster income growth are associated with lower values of $b_{ct}$, while protest activity is associated with higher values. More durable regimes and more democratic political systems also appear to reduce the attractiveness of coup activity. Despite the inclusion of these controls, the time trend remains statistically significant, suggesting that the long-run decline in coup attempts is only partially explained by the observable covariates included in the model.

The determinants of regime strength, $z_{ct}$, are estimated less precisely. The coup-proofing measure enters consistently and significantly, but its estimated effect is negative, implying that coup-proofing is more prevalent in observations associated with lower values of $z_{ct}$. One possible interpretation is that coup-proofing is a symptom: governments facing a greater risk of removal may adopt coup-proofing institutions in an effort to compensate for an unfavorable strategic position.

An important feature of the estimation results is their stability. We consider a variety of expanded specifications in Appendix \ref{simprob}, including regional and decadal fixed effects, and specifications in which exclusion restrictions for nonparametric identification are dropped. Across all alternative specifications considered, the principal coefficients change little in either sign or magnitude, despite substantial changes in the sets of identifying restrictions and nuisance controls. This stability suggests that the estimated decomposition reflects persistent features of the data rather than particular modeling choices. The results also suggest that broad macroeconomic variables account for much of the temporal and regional variation in coup incidence, as decadal and regional fixed effects add little explanatory power to the model.

\subsection{Full models}

We now turn to estimation of the full model, in which the information parameter is allowed to vary and $\sigma_{ct}>0$. As discussed above, identification of $\sigma_{ct}$ requires information on the extent of participation in coup attempts. Direct measures of participation are typically unavailable, but we can exploit information on the nature of coup attempts recorded in our data. Specifically, the dataset classifies coups according to a number of characteristics, including whether an attempt was conspiratorial in nature and whether it involved substantial popular participation. We interpret conspiratorial coups as weak attacks and popular coups as strong attacks, thereby obtaining a coarse measure of coup size. In context, this means that if an attack was noted to be conspiratorial, we posit that $A_{ct}<\eta_w=.1$. If an attack, however, was popular, we posit that $A_{ct}>\eta_s=.9$. In terms of the data, applying this classification yields the joint distribution of coup outcomes and attack strengths reported in Table \ref{tab:coup_strength}.

\begin{table}[!htbp]
\centering
\footnotesize
\caption{Attack strength by coup outcome}\vspace{.3cm}
\label{tab:coup_strength}
\begin{tabular}{lccc}
\toprule
Attack strength & Failed coup & Successful coup & Total \\
\midrule
Weak &       198 (100.0\%) &         0 ( 0.0\%) &       198 (100.0\%) \\
Neither &       254 (45.4\%) &       305 (54.6\%) &       559 (100.0\%) \\
Strong &        11 (13.4\%) &        71 (86.6\%) &        82 (100.0\%) \\
\midrule
Total &       463 &       376 &       839 \\
\bottomrule
\end{tabular}

\end{table}

Table \ref{tab:coup_strength} provides the empirical analogue of Figure \ref{parmranges}. The resulting classification reveals a clear, though imperfect, relationship between participation and success. Among failed coups, 198 are classified as weak and only 11 as strong. Conversely, none of the successful coups are classified as weak, while 71 are classified as strong. A substantial number of coups fall into neither category, indicating that participation is strongly associated with success but is not determinative.

Table \ref{tab:table2} reports estimates of the full model, in which the information parameter $\sigma_{ct}$ is estimated jointly with the coup-benefit and regime-strength equations. Table \ref{tab:table2} again provides a useful calibration benchmark. The estimated benefit intercept implies an unconditional coup-attempt probability of $\Phi(-1.402)=0.080$, which again matches the observed frequency of coup attempts in the data. The estimated log signal standard deviation is $0.0189$, implying $\hat{\sigma}=\exp(0.0189)=1.019$, which lies well below the uniqueness threshold $\sqrt{2\pi}\approx2.51$ required by Assumption~\ref{sigass}. Thus, even in the absence of observable covariates, the model implies a stable equilibrium information structure consistent with the theoretical framework.

The rest of the results reported on Table \ref{tab:table2} largely reaffirm the findings from the restricted models reported in Table \ref{tab:table1}. Economic performance, regime durability, and political institutions continue to exert their primary influence through the coup-benefit equation, $b_{ct}$, while coup-proofing measures remain important predictors of regime strength, $z_{ct}$. The principal innovation of the full model is the estimation of the information equation governing $\sigma_{ct}$.

An important diagnostic is whether the estimated values of $\sigma_{ct}$ satisfy the uniqueness condition asserted in Assumption \ref{unique} when covariates are included. The estimated log signal standard deviations are typically around $0.56$, implying signal standard deviations of approximately $\exp(0.56)=1.75$. This value is comfortably below the threshold $\sqrt{2\pi}=2.51$, ensuring that the maintained condition for uniqueness is satisfied throughout the sample. The final column of Table \ref{tab:table2} presents a parsimonious specification obtained by eliminating variables that contribute little explanatory power. While such simplification should always be approached with caution, the resulting specification retains the key theoretical determinants of coup benefits, regime strength, and information quality while providing a tractable framework for simulation and counterfactual analysis.

\begin{table}[!htbp]
\centering
\footnotesize
\caption{Simulated maximum likelihood parameter estimates of the full model}
\label{tab:table2}
{
\def\sym#1{\ifmmode^{#1}\else\(^{#1}\)\fi}
\begin{tabular}{l*{5}{c}}
\toprule
            &\multicolumn{1}{c}{(1)}&\multicolumn{1}{c}{(2)}&\multicolumn{1}{c}{(3)}&\multicolumn{1}{c}{(4)}&\multicolumn{1}{c}{(5)}\\
            &\multicolumn{1}{c}{Constants}&\multicolumn{1}{c}{Time Trend}&\multicolumn{1}{c}{+ Proofing Excl.}&\multicolumn{1}{c}{Variance Trend}&\multicolumn{1}{c}{Preferred}\\
\midrule
\textbf{Coup benefits \(b_{ct}\)}&                     &                     &                     &                     &                     \\
Constant    &      -1.402\sym{***}&       6.979\sym{***}&       27.54\sym{***}&       21.48\sym{**} &       27.88\sym{***}\\
            &    (0.0177)         &     (1.587)         &     (4.597)         &     (9.730)         &     (4.542)         \\
Year        &                     &    -0.00423\sym{***}&     -0.0137\sym{***}&     -0.0105\sym{**} &     -0.0139\sym{***}\\
            &                     &  (0.000801)         &   (0.00230)         &   (0.00487)         &   (0.00227)         \\
Real GDP per capita growth&                     &                     &      -1.432\sym{***}&      -1.126\sym{**} &      -1.425\sym{***}\\
            &                     &                     &     (0.357)         &     (0.495)         &     (0.354)         \\
Log GDP per capita&                     &                     &      -0.181\sym{***}&      -0.245\sym{***}&      -0.179\sym{***}\\
            &                     &                     &    (0.0256)         &    (0.0377)         &    (0.0250)         \\
Regime durability&                     &                     &    -0.00651\sym{***}&    -0.00736\sym{**} &    -0.00623\sym{***}\\
            &                     &                     &   (0.00182)         &   (0.00303)         &   (0.00182)         \\
Polity Index&                     &                     &     -0.0126\sym{***}&     -0.0215\sym{***}&     -0.0106\sym{**} \\
            &                     &                     &   (0.00489)         &   (0.00686)         &   (0.00464)         \\
Population growth rate&                     &                     &      -1.622         &       0.313         &                     \\
            &                     &                     &     (2.056)         &     (2.822)         &                     \\
Log protests&                     &                     &                     &       0.190\sym{***}&                     \\
            &                     &                     &                     &    (0.0544)         &                     \\
\midrule
\textbf{Regime strength \(z_{ct}\)}&                     &                     &                     &                     &                     \\
Constant    &      0.0731\sym{***}&       0.668         &      -6.405         &       0.428         &       0.150         \\
            &    (0.0222)         &     (1.655)         &     (5.243)         &     (10.05)         &    (0.0983)         \\
Year        &                     &   -0.000298         &     0.00330         &  0.00000977         &                     \\
            &                     &  (0.000836)         &   (0.00261)         &   (0.00504)         &                     \\
Real GDP per capita growth&                     &                     &    -0.00555         &       0.596         &                     \\
            &                     &                     &     (0.438)         &     (0.686)         &                     \\
Log GDP per capita&                     &                     &     0.00311         &    -0.00965         &                     \\
            &                     &                     &    (0.0282)         &    (0.0415)         &                     \\
Regime durability&                     &                     &    -0.00160         &    -0.00558         &                     \\
            &                     &                     &   (0.00202)         &   (0.00373)         &                     \\
Polity Index&                     &                     &    -0.00816         &     -0.0182\sym{**} &                     \\
            &                     &                     &   (0.00499)         &   (0.00792)         &                     \\
Population growth rate&                     &                     &       7.561\sym{***}&       8.532\sym{**} &       7.302\sym{***}\\
            &                     &                     &     (2.510)         &     (3.743)         &     (2.683)         \\
Coup-proofing&                     &                     &      -0.133\sym{***}&      -0.175\sym{***}&      -0.112\sym{***}\\
            &                     &                     &    (0.0496)         &    (0.0632)         &    (0.0426)         \\
\midrule
\textbf{Log signal dispersion \(\ln\sigma_{ct}\)}&                     &                     &                     &                     &                     \\
Constant    &      0.0189         &      0.0153         &     -0.0304         &      -33.74\sym{*}  &      -22.49\sym{***}\\
            &    (0.0271)         &    (0.0262)         &    (0.0476)         &     (17.39)         &     (6.826)         \\
Year        &                     &                     &                     &      0.0169\sym{*}  &      0.0113\sym{***}\\
            &                     &                     &                     &   (0.00867)         &   (0.00342)         \\
Information reliability&                     &                     &                     &      -0.208         &      -0.104\sym{*}  \\
            &                     &                     &                     &     (0.136)         &    (0.0572)         \\
\midrule
\(N\)       &       10764         &       10764         &        5076         &        3243         &        5076         \\
\bottomrule
\multicolumn{6}{l}{\footnotesize Notes: Estimation follows equation (\ref{full_ll}) with $K=200$.}\\
\multicolumn{6}{l}{\footnotesize All explanatory variables are lagged one year.}\\
\multicolumn{6}{l}{\footnotesize The contraction mapping in equation (\ref{cmap}) is iterated until average absolute value error is less than 1e-10.}\\
\multicolumn{6}{l}{\footnotesize *, **, and *** denote significance at the 10\%, 5\%, and 1\% levels.}\\
\end{tabular}
}

\end{table}

In Appendix \ref{fullrob}, we report alternative specifications in which regional and decadal fixed effects are considered, and specifications in which exclusion restrictions are relaxed. As as the case with the simple models, across all alternative specifications, the estimated desirability and regime-strength equations remain remarkably stable. 

Appendix \ref{etatests} reports robustness checks using alternative definitions of weak and strong attacks, varying the threshold values of $\eta_w$ and $\eta_s$. As expected, the estimated parameters exhibit some sensitivity to these classification rules. More stringent definitions of weak and strong attacks (e.g., $\eta_w=.05$ and $\eta_s=.95$) generally lead to larger estimates of $\sigma_{ct}$, while wider thresholds result in less precise estimates of the information equation. Nevertheless, although the information parameters display some sensitivity to these alternative specifications, the principal structural decomposition remains robust.

\section{Counterfactuals} \label{app}

As a way of exploring the implications of the estimated model, we conduct a series of counterfactual exercises. We do not view the resulting counterfactuals as providing definitive conclusions about the determinants of coups. Rather, they are intended as illustrations of how the model may be used to interpret the history of political upheaval over the past several decades and to conduct simple thought experiments regarding alternative historical trajectories. Throughout this section, we use the parsimonious specification reported in the final column of Table \ref{tab} as our baseline model for simulations

The simulations proceed as follows. Given the estimated equations for coup benefits, regime strength, and signal dispersion, together with the observed history of coup occurrence, coup success, and coup strength (weak, intermediate, or strong), we recover latent disturbances corresponding to the three equations of the model. Specifically, we draw latent values for $\epsilon_b$, $\epsilon_z$, and $\epsilon_s$ that are consistent with both the estimated linear indices and the observed coup outcomes. We generate $S=200$ such latent histories. These recovered latent histories provide a foundation for conducting counterfactual exercises while preserving the structure of the observed data.

We begin with a simple thought experiment that operates through the coup-benefit equation. Specifically, we ask how the distribution of coups might have differed in a more democratic world. Within the model, this is equivalent to asking how coup outcomes would change if the average benefits from a successful coup were reduced. To examine this question, we increase the Polity Index by one sample standard deviation. Using the estimates reported in Table \ref{tab:table2}, this change lowers the coup-benefit index by approximately $\Delta b_{ct} = \beta_{\text{polity}}\Delta\textit{Polity} $ where $\beta_{\text{polity}}<0$. Table \ref{tab} reports the resulting changes in the frequency and composition of coups.

\begin{table}[!htbp]
\centering
\footnotesize
\caption{Counterfactual -- A One-Standard-Deviation Increase in the Polity Index}
\label{tab}
\begin{tabular}{lccc}
\toprule
Outcome & Observed & Counterfactual & \% Change \\
\midrule
No coups &      4678 &   4706.07 (     4.70) &      0.60 (     0.10) \\
Coups &       398 &    369.93 (     4.70) &     -7.05 (     1.18) \\
Failed coups &       225 &    207.39 (     5.65) &     -7.83 (     2.51) \\
Successful coups &       173 &    162.54 (     3.14) &     -6.04 (     1.82) \\
\bottomrule
\end{tabular}
\vspace{0.1in}
\begin{flushleft}
\footnotesize Standard deviations and means are computed across 200 draws of latent factors consistent with the observed frequency and nature of coups, and the estimates presented in column 5 of table \ref{tab:table2}.
\end{flushleft}

\end{table}

The results suggest that a more democratic world would have experienced fewer coups. Interpreted within the model framework, this occurs because greater democracy reduces the expected benefits of participating in a coup, thereby lowering the frequency with which coup attempts are undertaken; not, for example, because more democratic regimes are inherently stronger. Quantitatively, the simulations imply that the number of coups would have fallen by approximately 7 percent, with somewhat larger reductions in failed coups than in successful coups.

We hasten to add that this exercise should be interpreted as a partial-equilibrium counterfactual. The only channel through which democracy operates is the direct effect captured in the estimated coup-benefit equation. We do not allow changes in democracy to affect other determinants of coups. In practice, however, greater democracy may influence economic growth, information quality, regime durability, or other variables that also enter the model. To the extent that such relationships exist, the overall effect of democratization would include additional indirect channels beyond those considered here. Incorporating these broader interactions would require a substantially richer model in which our model of coup activity is nested in an encompassing general equilibrium model of economic activity, and is therefore left for future research.

Table \ref{tab} presents the results of a second counterfactual exercise. As with the previous experiment, the objective is not to provide a definitive estimate of the effects of a particular policy, but rather to illustrate the distinct role played by regime strength in the model. To do so, we consider a one-standard-deviation increase in the coup-proofing variable. Using the estimates reported in Table \ref{tab:table2}, this change alters the regime-strength index by approximately
$\Delta z_{ct} = \zeta_{\text{proofing}} \, \Delta \textit{Coup-Proofing}$ 
where the estimated coefficient on coup-proofing is negative, implying a reduction in the average regime-strength index. Consequently, the counterfactual generates fewer failed coups and more successful coups. 

\begin{table}[!htbp]
\centering
\footnotesize
\caption{Counterfactual -- Coup-Proofing}
\label{tabcp}
\begin{tabular}{lccc}
\toprule
Outcome & Observed & Counterfactual & \% Change \\
\midrule
No coups &      4678 &   4678.00 (     0.00) &      0.00 (     0.00) \\
Coups &       398 &    398.00 (     0.00) &      0.00 (     0.00) \\
Failed coups &       225 &    198.13 (     4.26) &    -11.94 (     1.89) \\
Successful coups &       173 &    199.87 (     4.26) &     15.53 (     2.46) \\
\bottomrule
\end{tabular}
\vspace{0.1in}
\begin{flushleft}
\footnotesize Notes: Counterfactual increases coup-proofing by one sample standard deviation. Counterfactual values are means across latent-history draws; standard deviations across draws are in parentheses.
\end{flushleft}

\end{table}

The results differ markedly from those of the democracy counterfactual. Whereas changes in the coup-benefit equation primarily affect the frequency of coup attempts, changes in regime strength primarily affect the probability that a given coup succeeds. As expected from the structure of the model, decreasing regime strength has essentially no effect on the number of coup attempts. Instead, it changes the composition of observed coups by shifting some failed  coups into the successful category.

Quantitatively, the simulations suggest that weaker regimes would have generated approximately 15 percent more successful coups and a corresponding reduction in failed coups. Put differently, a number of coups that failed in the observed data would instead have succeeded had regimes possessed greater underlying strength, but the number of coup attempts would not have changed. This distinction highlights the role of the regime-strength equation in the model: rather than influencing whether citizens attempt a coup, it primarily influences whether the regime can withstand an attempt once it occurs.

As with the previous exercise, this counterfactual should be interpreted as a partial-equilibrium thought experiment. We alter only the regime-strength index while holding all other determinants fixed. In reality, coup-proofing measures may themselves be related to deeper political and economic conditions that also influence coup activity. Indeed, the negative relationship between coup-proofing and estimated regime strength reported in Table \ref{tabcp} suggests that governments may adopt coup-proofing institutions in response to underlying vulnerability. Incorporating such feedback mechanisms would require a richer model of the joint determination of coup-proofing and regime stability and is beyond the scope of the present analysis. 

A final counterfactual focuses on the third index in the model, namely the information equation governing $\sigma_{ct}$. As before, we ask a simple question: how might the history of coups have differed if citizens had, on average, access to more reliable information? To examine this possibility, we consider a one-standard-deviation increase in the information reliability index. Using the estimates reported in Table \ref{tab:table2}, this change reduces the log signal standard deviation by approximately $\Delta \ln \sigma_{ct}=\beta_{\mathrm{info}}\Delta \textit{Information\ Reliability}$ thereby increasing the precision of private information available to potential coup participants. The results of this exercise are reported in Table \ref{tab:tableir}, where we have now reported how the distribution of coups across success, failure, and attack strength might change under the counterfactual. 

\begin{table}[!htbp]
\centering
\footnotesize
\caption{Counterfactual -- Information Reliability}
\label{tab:tableir}
\begin{tabular}{lccc}
\toprule
Outcome & Observed & Counterfactual & \% Change \\
\midrule
No coups &      4678 &   4678.00 (     0.00) &      0.00 (     0.00) \\
Coups &       398 &    398.00 (     0.00) &      0.00 (     0.00) \\
Failed coups &       225 &    224.59 (     1.29) &     -0.18 (     0.57) \\
Successful coups &       173 &    173.41 (     1.29) &      0.24 (     0.74) \\
\addlinespace
Weak &        98 &    100.96 (     2.11) &      3.02 (     2.15) \\
Neither &       261 &    256.95 (     3.03) &     -1.55 (     1.16) \\
Strong &        39 &     40.09 (     2.06) &      2.79 (     5.28) \\
\addlinespace
\midrule
Strong failures &         3 &      1.21 (     0.84) &    -59.67 (    28.08) \\
Weak successes &         0 &      0.17 (     0.38) & -- \\
Coordination error rate (\%) &      0.75 &      0.35 (     0.23) &    -53.83 (    30.20) \\
\bottomrule
\end{tabular}
\vspace{0.1in}
\begin{flushleft}
\footnotesize Notes: Counterfactual increases the information reliability index by one sample standard deviation. This changes \(x_s\beta_s\), where \(\sigma=\exp(x_s\beta_s)\). Counterfactual values are means across latent-history draws; standard deviations across draws are in parentheses. Strong failures and weak successes are coordination-error cross-cases.
\end{flushleft}

\end{table}

As in the coup-proofing counterfactual, changes to the information equation have little effect on the overall incidence of coups. This is consistent with the structure of the model, in which the frequency of coup attempts is determined primarily by the coup-benefit equation. Nor do improvements in information quality substantially alter the overall success rate of coups. Instead, the principal effect of better information is to alter the composition of coup outcomes.

In particular, the simulations summarized in Table \ref{tab:tableir} suggest that improved information reduces the frequency of what might be termed coordination failures. The number of strong coups that fail falls from approximately three in the observed data to roughly one under the counterfactual. Although quantitatively modest, this result is theoretically significant. Strong failed coups correspond to situations in which participation is high even though the coup ultimately fails, implying that a substantial number of citizens incur the costs of participation without achieving regime change. From the perspective of the model, these outcomes represent particularly severe coordination failures.

Viewed through the lens of the simple model presented earlier, one might think of these cases as deviations from a first-best benchmark. In that benchmark, successful coups would be accompanied by widespread participation, while failed coups would attract little or no participation. Better information moves outcomes closer to this benchmark by reducing the frequency with which citizens coordinate on participation levels that are poorly aligned with the underlying strength of the regime. In this sense, the information counterfactual highlights a distinct role for information quality: rather than changing how often coups occur, it improves the efficiency with which citizens coordinate their participation decisions.

As with the previous counterfactuals, this exercise should be interpreted as a partial-equilibrium thought experiment. We allow information quality to change while holding all other determinants of coup activity fixed. In practice, information quality may be correlated with broader political and economic developments that also influence coup incentives and regime strength. Modeling these interactions is an important topic for future work but lies beyond the scope of the present analysis.

\section{Conclusion} \label{conc}

Understanding why coups occur remains a central question in the study of political instability. At the same time, coups present an inherently strategic problem, as their success depends upon the ability of citizens, military officers, and political actors to coordinate their actions in the face of uncertainty. While economists and political scientists have developed rich theoretical models of coup activity, comparatively little work has sought to estimate such models structurally using cross-country data. In this paper, we have sought to bridge that gap by developing and estimating a structural model of coups derived from a global-games framework.

A central feature of the approach is the decomposition of coup activity into three distinct components: the benefits of participating in a successful coup, the strength of the regime against which the coup is directed, and the quality of information available to potential participants. This decomposition provides a natural framework for interpreting many of the variables commonly employed in the empirical coups literature. More importantly, it requires the researcher to think carefully about which factors affect the desirability of a coup, which affect its feasibility, and which influence the ability of citizens to coordinate their actions. Empirically, this decomposition relies on a combination of equilibrium restrictions, normalization assumptions, and exclusion restrictions that distinguish the desirability and feasibility dimensions of regime change. As with other structural latent-variable models, this decomposition is identified under a set of maintained equilibrium, normalization, and exclusion assumptions.

The counterfactual exercises illustrate the value of this decomposition. Because the model separately identifies coup desirability, regime strength, and information quality, it becomes possible to distinguish policies that primarily affect the frequency of coups from those that influence the likelihood of success or the efficiency of coordination among potential participants.

The framework also admits a wide range of extensions, including richer measures of participation, dynamic strategic interactions, cross-country spillovers, and more systematic methods for distinguishing determinants of coup desirability from those of regime strength. Because the model is fully structural, it also lends itself naturally to a wide range of counterfactual exercises. One might examine the consequences of alternative Cold War trajectories, different paths of democratization, changes in information technology, or the evolution of state capacity in developing countries. Such questions are difficult to address using reduced-form approaches alone, but arise naturally within the framework developed here.

Ultimately, coups are neither purely political nor purely economic phenomena; they emerge from the interaction of incentives, institutions, and information. By providing a framework in which these forces can be modeled jointly and applied to the data, we hope this framework offers a useful foundation for future work on the causes, consequences, and prevention of political upheaval.

\newpage
\begin{appendices} \label{appendix}

\counterwithin{table}{section}
\renewcommand{\thetable}{\thesection.\arabic{table}}

\section{Robustness}
In this appendix, we examine the robustness of our baseline specification, using the preferred models reported in Tables \ref{tab:table1} and \ref{tab:table2} as benchmarks. The first two tables investigate the sensitivity of the estimates to several alternative specifications. These include relaxing the exclusion restrictions used to achieve nonparametric identification, thereby relying instead on the index structure of the model for identification, as well as introducing alternative regional and temporal controls. Although the nature of the data precludes the inclusion of full country and year fixed effects, these specifications are intended to approximate such controls where possible. The principal conclusion is that the baseline specification proves remarkably robust to these alternative assumptions.

The final set of robustness checks examines the assumptions linking attack size to the model. Our baseline specification classifies attacks with participation rates of at least 90\% as strong and those with participation rates of at most 10\% as weak. We therefore re-estimate the model under alternative threshold values for these classifications. As expected, the estimated information parameters exhibit some sensitivity to these assumptions, but the underlying structural decomposition and the substantive counterfactual conclusions remain highly robust.

\subsection{Simple model} \label{simprob}

To examine the sensitivity of the simple model to alternative identifying assumptions, Table \ref{tab} reports a series of robustness exercises. Column (1) reproduces the preferred specification, while Columns (2)–(5) successively relax the exclusion restrictions and introduce alternative control structures. Specifically, Column (2) allows coup-proofing measures to enter both the coup-benefit and regime-strength equations, Column (3) allows protests to enter both equations, Column (4) includes World Bank regional fixed effects, and Column (5) replaces the linear time trend with decade fixed effects.

Across these alternative specifications, the principal parameter estimates remain remarkably stable in both sign and magnitude. In particular, the estimated effects of economic performance, regime durability, and democracy on coup desirability, together with the effects of population growth and coup-proofing on regime strength, vary little across specifications. These results suggest that the empirical decomposition into coup desirability and regime strength is not an artifact of the identifying restrictions or the particular treatment of regional and temporal effects.

\begin{table}[!htbp]
\centering
\scriptsize
\caption{Estimated models of mean benefits and regime strength}
\vspace{.3cm}
\label{tab:tableA1}
{
\def\sym#1{\ifmmode^{#1}\else\(^{#1}\)\fi}
\begin{tabular}{l*{5}{c}}
\toprule
            &\multicolumn{1}{c}{(1)}&\multicolumn{1}{c}{(2)}&\multicolumn{1}{c}{(3)}&\multicolumn{1}{c}{(4)}&\multicolumn{1}{c}{(5)}\\
            &\multicolumn{1}{c}{Preferred}&\multicolumn{1}{c}{Coup Proofing}&\multicolumn{1}{c}{Ln Protests}&\multicolumn{1}{c}{Region FE}&\multicolumn{1}{c}{Decade FE}\\
\midrule
\textbf{Coup benefits $b_{ct}$}&                     &                     &                     &                     &                     \\
Year        &     -0.0136\sym{***}&     -0.0137\sym{***}&     -0.0105\sym{**} &     -0.0136\sym{***}&                     \\
            &   (0.00231)         &   (0.00233)         &   (0.00484)         &   (0.00248)         &                     \\
Real GDP per capita growth&      -1.427\sym{***}&      -1.414\sym{***}&      -1.122\sym{**} &      -1.497\sym{***}&      -1.499\sym{***}\\
            &     (0.356)         &     (0.356)         &     (0.498)         &     (0.364)         &     (0.360)         \\
Log GDP per capita&      -0.175\sym{***}&      -0.176\sym{***}&      -0.240\sym{***}&      -0.118\sym{***}&      -0.178\sym{***}\\
            &    (0.0254)         &    (0.0257)         &    (0.0382)         &    (0.0338)         &    (0.0258)         \\
Regime durability&    -0.00646\sym{***}&    -0.00650\sym{***}&    -0.00738\sym{**} &    -0.00633\sym{***}&    -0.00669\sym{***}\\
            &   (0.00185)         &   (0.00185)         &   (0.00302)         &   (0.00186)         &   (0.00187)         \\
Polity Index&     -0.0119\sym{**} &     -0.0130\sym{***}&     -0.0216\sym{***}&     -0.0174\sym{***}&     -0.0131\sym{***}\\
            &   (0.00472)         &   (0.00489)         &   (0.00693)         &   (0.00533)         &   (0.00489)         \\
Constant    &       27.29\sym{***}&       27.49\sym{***}&       21.37\sym{**} &       26.78\sym{***}&       0.506\sym{**} \\
            &     (4.607)         &     (4.641)         &     (9.667)         &     (4.927)         &     (0.230)         \\
Population growth rate&                     &      -1.996         &      0.0916         &      -2.986         &      -2.044         \\
            &                     &     (2.086)         &     (2.854)         &     (2.459)         &     (2.086)         \\
Coup-proofing measures&                     &     -0.0128         &                     &                     &                     \\
            &                     &    (0.0439)         &                     &                     &                     \\
Log protests&                     &                     &       0.174\sym{***}&                     &                     \\
            &                     &                     &    (0.0556)         &                     &                     \\
\midrule
\textbf{Regime strength $z_{ct}$}&                     &                     &                     &                     &                     \\
Year        &      0.0178\sym{***}&      0.0166\sym{***}&      0.0173         &      0.0133\sym{**} &                     \\
            &   (0.00560)         &   (0.00608)         &    (0.0132)         &   (0.00650)         &                     \\
Population growth rate&       18.41\sym{***}&       18.94\sym{***}&       14.39\sym{*}  &       17.97\sym{***}&       18.55\sym{***}\\
            &     (5.871)         &     (6.030)         &     (8.270)         &     (6.946)         &     (6.017)         \\
Coup-proofing measures&      -0.178         &      -0.184\sym{*}  &      -0.378\sym{**} &      -0.159         &      -0.204\sym{*}  \\
            &     (0.110)         &     (0.111)         &     (0.161)         &     (0.116)         &     (0.112)         \\
Constant    &      -35.08\sym{***}&      -32.58\sym{***}&      -33.27         &      -26.57\sym{**} &     -0.0254         \\
            &     (11.12)         &     (12.19)         &     (26.24)         &     (12.95)         &     (0.646)         \\
Real GDP per capita growth&                     &       0.101         &       0.868         &       0.243         &       0.192         \\
            &                     &     (0.931)         &     (1.462)         &     (0.954)         &     (0.937)         \\
Log GDP per capita&                     &    -0.00709         &     -0.0187         &      0.0252         &     0.00305         \\
            &                     &    (0.0679)         &     (0.104)         &    (0.0915)         &    (0.0686)         \\
Regime durability&                     &     0.00272         &   -0.000503         &     0.00238         &     0.00219         \\
            &                     &   (0.00401)         &   (0.00862)         &   (0.00410)         &   (0.00408)         \\
Polity Index&                     &     0.00710         &     0.00311         &     0.00899         &     0.00760         \\
            &                     &    (0.0131)         &    (0.0200)         &    (0.0136)         &    (0.0131)         \\
Log protests&                     &                     &      -0.181         &                     &                     \\
            &                     &                     &     (0.143)         &                     &                     \\
\midrule
\(N\)       &        5076         &        5076         &        3243         &        5076         &        5076         \\
\bottomrule
\multicolumn{6}{l}{\footnotesize Notes: Estimation follows equation (\ref{simplell}) with $K=100$.}\\
\multicolumn{6}{l}{\footnotesize All explanatory variables are lagged one year.}\\
\multicolumn{6}{l}{\footnotesize *, **, and *** denote significance at the 10\%, 5\%, and 1\% levels.}\\
\end{tabular}
}

\end{table}

\subsection{Full model} \label{fullrob}

Table \ref{tab:tableA2} reports analogous robustness exercises for the full structural model. As before, Column (1) presents the preferred specification, while Columns (2)–(5) successively relax the exclusion restrictions and introduce alternative regional and temporal controls. Overall, the estimated structural parameters remain stable across specifications. The determinants of coup desirability exhibit nearly identical signs and magnitudes throughout, while the estimated effects of population growth and coup-proofing on regime strength vary little despite substantial changes in specification. Estimates of the signal-dispersion equation display somewhat greater sensitivity, reflecting the comparatively weaker information available to identify the information structure, but the qualitative conclusions remain unchanged. Taken together, these results suggest that the empirical decomposition into coup desirability, regime strength, and information precision is not an artifact of the identifying restrictions or the particular treatment of regional and temporal effects.

\begin{table}[!htbp]
\centering
\scriptsize
\caption{Estimated models of mean benefits and regime strength}
\vspace{.3cm}
\label{tab:tableA2}
{
\def\sym#1{\ifmmode^{#1}\else\(^{#1}\)\fi}
\begin{tabular}{l*{5}{c}}
\toprule
            &\multicolumn{1}{c}{(1)}&\multicolumn{1}{c}{(2)}&\multicolumn{1}{c}{(3)}&\multicolumn{1}{c}{(4)}&\multicolumn{1}{c}{(5)}\\
            &\multicolumn{1}{c}{Preferred}&\multicolumn{1}{c}{Coup-proofing}&\multicolumn{1}{c}{Ln Protests}&\multicolumn{1}{c}{Region FE}&\multicolumn{1}{c}{Decade FE}\\
\midrule
\textbf{Coup benefits $b_{ct}$}&                     &                     &                     &                     &                     \\
Year        &     -0.0139\sym{***}&     -0.0136\sym{***}&     -0.0104\sym{**} &     -0.0135\sym{***}&                     \\
            &   (0.00227)         &   (0.00230)         &   (0.00475)         &   (0.00246)         &                     \\
Real GDP per capita growth&      -1.425\sym{***}&      -1.419\sym{***}&      -1.138\sym{**} &      -1.501\sym{***}&      -1.500\sym{***}\\
            &     (0.354)         &     (0.354)         &     (0.494)         &     (0.358)         &     (0.358)         \\
Log GDP per capita&      -0.179\sym{***}&      -0.180\sym{***}&      -0.244\sym{***}&      -0.122\sym{***}&      -0.182\sym{***}\\
            &    (0.0250)         &    (0.0255)         &    (0.0377)         &    (0.0335)         &    (0.0257)         \\
Regime durability&    -0.00623\sym{***}&    -0.00650\sym{***}&    -0.00740\sym{**} &    -0.00632\sym{***}&    -0.00672\sym{***}\\
            &   (0.00182)         &   (0.00184)         &   (0.00300)         &   (0.00185)         &   (0.00185)         \\
Polity Index&     -0.0106\sym{**} &     -0.0126\sym{***}&     -0.0211\sym{***}&     -0.0170\sym{***}&     -0.0127\sym{***}\\
            &   (0.00464)         &   (0.00484)         &   (0.00684)         &   (0.00528)         &   (0.00484)         \\
Constant    &       27.88\sym{***}&       27.38\sym{***}&       21.28\sym{**} &       26.67\sym{***}&       0.544\sym{**} \\
            &     (4.542)         &     (4.593)         &     (9.493)         &     (4.889)         &     (0.229)         \\
Population growth rate&                     &      -1.617         &       0.229         &      -2.590         &      -1.675         \\
            &                     &     (2.056)         &     (2.801)         &     (2.429)         &     (2.064)         \\
Coup-proofing measures&                     &     -0.0125         &                     &                     &                     \\
            &                     &    (0.0434)         &                     &                     &                     \\
Log protests&                     &                     &       0.174\sym{***}&                     &                     \\
            &                     &                     &    (0.0550)         &                     &                     \\
\midrule
\textbf{Regime strength $z_{ct}$}&                     &                     &                     &                     &                     \\
Population growth rate&       7.302\sym{***}&       6.623\sym{**} &       8.274\sym{**} &       3.454         &       5.507\sym{*}  \\
            &     (2.683)         &     (2.897)         &     (3.636)         &     (3.437)         &     (2.893)         \\
Coup-proofing measures&      -0.112\sym{***}&      -0.131\sym{**} &      -0.145\sym{**} &      -0.109\sym{**} &      -0.140\sym{***}\\
            &    (0.0426)         &    (0.0545)         &    (0.0660)         &    (0.0520)         &    (0.0542)         \\
Constant    &       0.150         &      -6.084         &       0.606         &      -1.059         &       0.183         \\
            &    (0.0983)         &     (6.093)         &     (10.97)         &     (6.227)         &     (0.285)         \\
Year        &                     &     0.00321         &  -0.0000850         &    0.000334         &                     \\
            &                     &   (0.00305)         &   (0.00550)         &   (0.00313)         &                     \\
Real GDP per capita growth&                     &      -0.116         &       0.235         &      -0.278         &      -0.168         \\
            &                     &     (0.495)         &     (0.683)         &     (0.443)         &     (0.466)         \\
Log GDP per capita&                     &     -0.0114         &     0.00809         &      0.0734\sym{*}  &    -0.00480         \\
            &                     &    (0.0328)         &    (0.0425)         &    (0.0428)         &    (0.0295)         \\
Regime durability&                     &    -0.00135         &    -0.00583         &    -0.00141         &    -0.00175         \\
            &                     &   (0.00214)         &   (0.00441)         &   (0.00211)         &   (0.00220)         \\
Polity Index&                     &    -0.00722         &     -0.0141         &    -0.00540         &    -0.00728         \\
            &                     &   (0.00632)         &   (0.00915)         &   (0.00650)         &   (0.00608)         \\
Log protests&                     &                     &      -0.130\sym{**} &                     &                     \\
            &                     &                     &    (0.0634)         &                     &                     \\
\midrule
\textbf{Log signal dispersion \(\ln\sigma_{ct}\)}&                     &                     &                     &                     &                     \\
Year        &      0.0113\sym{***}&     0.00874\sym{**} &      0.0173\sym{*}  &     0.00894\sym{**} &                     \\
            &   (0.00342)         &   (0.00429)         &   (0.00909)         &   (0.00386)         &                     \\
Information reliability&      -0.104\sym{*}  &     -0.0681         &      -0.188         &     -0.0461         &     -0.0812         \\
            &    (0.0572)         &    (0.0626)         &     (0.141)         &    (0.0650)         &    (0.0656)         \\
Constant    &      -22.49\sym{***}&      -17.44\sym{**} &      -34.54\sym{*}  &      -17.84\sym{**} &      -0.125         \\
            &     (6.826)         &     (8.529)         &     (18.22)         &     (7.689)         &     (0.114)         \\
\midrule
\(N\)       &        5076         &        5076         &        3243         &        5076         &        5076         \\
\bottomrule
\multicolumn{6}{l}{\footnotesize Notes: Estimation follows equation (\ref{simplell}) with $K=200$.}\\
\multicolumn{6}{l}{\footnotesize All explanatory variables are lagged one year.}\\
\multicolumn{6}{l}{\footnotesize *, **, and *** denote significance at the 10\%, 5\%, and 1\% levels.}\\
\end{tabular}
}

\end{table}

\subsection{Attack strength} \label{etatests}

Table \ref{tab:tableA3} examines the sensitivity of the full model to alternative definitions of weak and strong coups by varying the participation thresholds used to classify observed events. Across all specifications, the estimated coup-desirability and regime-strength equations remain stable, with both the signs and magnitudes of the principal coefficients exhibiting little variation. In contrast, the signal-dispersion equation displays greater sensitivity to the choice of participation thresholds. As the definitions of weak and strong coups become less restrictive, the estimated effect of information reliability attenuates and eventually loses statistical significance, while retaining its predicted negative sign. This pattern is consistent with the role of the participation thresholds in identifying the information process: broader thresholds admit observations that are less clearly indicative of weak or strong coordination, thereby reducing the precision with which signal dispersion can be estimated. Nevertheless, the underlying structural decomposition and the counterfactual conclusions that follow from it remain highly robust to reasonable alternative classifications of participation.

\begin{table}[!htbp]
\centering
\scriptsize
\caption{Estimated models of mean benefits and regime strength}
\vspace{.3cm}
\label{tab:tableA3}
{
\def\sym#1{\ifmmode^{#1}\else\(^{#1}\)\fi}
\begin{tabular}{l*{4}{c}}
\toprule
            &\multicolumn{1}{c}{(1)}&\multicolumn{1}{c}{(2)}&\multicolumn{1}{c}{(3)}&\multicolumn{1}{c}{(4)}\\
            &\multicolumn{1}{c}{Baseline}&\multicolumn{1}{c}{$\eta\_w=.05, \eta\_s=.95$}&\multicolumn{1}{c}{$\eta\_w=.20, \eta\_s=.80$}&\multicolumn{1}{c}{$\eta\_w=.25, \eta\_s=.75$}\\
\midrule
\textbf{Coup benefits $b\_{ct}$}&                     &                     &                     &                     \\
Year        &     -0.0139\sym{***}&     -0.0139\sym{***}&     -0.0139\sym{***}&     -0.0139\sym{***}\\
            &   (0.00225)         &   (0.00228)         &   (0.00233)         &   (0.00223)         \\
Real GDP per capita growth&      -1.409\sym{***}&      -1.448\sym{***}&      -1.409\sym{***}&      -1.397\sym{***}\\
            &     (0.345)         &     (0.352)         &     (0.393)         &     (0.331)         \\
Log GDP per capita&      -0.179\sym{***}&      -0.178\sym{***}&      -0.178\sym{***}&      -0.178\sym{***}\\
            &    (0.0249)         &    (0.0253)         &    (0.0254)         &    (0.0246)         \\
Regime durability&    -0.00630\sym{***}&    -0.00624\sym{***}&    -0.00646\sym{***}&    -0.00641\sym{***}\\
            &   (0.00180)         &   (0.00181)         &   (0.00181)         &   (0.00178)         \\
Polity Index&     -0.0110\sym{**} &     -0.0105\sym{**} &     -0.0112\sym{**} &     -0.0112\sym{**} \\
            &   (0.00460)         &   (0.00459)         &   (0.00476)         &   (0.00455)         \\
Constant    &       28.02\sym{***}&       27.87\sym{***}&       27.98\sym{***}&       27.87\sym{***}\\
            &     (4.492)         &     (4.565)         &     (4.645)         &     (4.453)         \\
\midrule
\textbf{Regime strength $z\_{ct}$}&                     &                     &                     &                     \\
Population growth rate&       7.005\sym{***}&       8.804\sym{**} &       5.641\sym{**} &       5.047\sym{**} \\
            &     (2.674)         &     (3.892)         &     (2.684)         &     (2.466)         \\
Coup-proofing measures&      -0.127\sym{***}&      -0.123\sym{***}&      -0.100\sym{**} &     -0.0949\sym{**} \\
            &    (0.0430)         &    (0.0454)         &    (0.0461)         &    (0.0476)         \\
Constant    &       0.201\sym{*}  &       0.145         &       0.166\sym{*}  &       0.160\sym{*}  \\
            &     (0.104)         &     (0.118)         &     (0.100)         &    (0.0974)         \\
\midrule
\textbf{Log signal dispersion \(\ln\sigma\_{ct}\)}&                     &                     &                     &                     \\
Year        &     0.00817\sym{**} &      0.0193\sym{***}&     0.00458         &     0.00381         \\
            &   (0.00381)         &   (0.00611)         &   (0.00374)         &   (0.00335)         \\
Information reliability&     -0.0930\sym{*}  &      -0.153\sym{**} &     -0.0242         &     -0.0189         \\
            &    (0.0533)         &    (0.0599)         &    (0.0659)         &    (0.0577)         \\
Constant    &      -16.30\sym{**} &      -38.41\sym{***}&      -9.193         &      -7.682         \\
            &     (7.589)         &     (12.20)         &     (7.455)         &     (6.659)         \\
\midrule
\(N\)       &        5076         &        5076         &        5076         &        5076         \\
\bottomrule
\multicolumn{5}{l}{\footnotesize Notes: Estimation follows equation (\ref{simplell}) with $K=100$.}\\
\multicolumn{5}{l}{\footnotesize All explanatory variables are lagged one year.}\\
\multicolumn{5}{l}{\footnotesize *, **, and *** denote significance at the 10\%, 5\%, and 1\% levels.}\\
\end{tabular}
}

\end{table}

\end{appendices}

\clearpage
\newpage
\singlespacing
\bibliographystyle{chicago} 
\bibliography{sources} 

@book{AbiddeKumahAbiwu2023,
  author = {Abidde, Sabella O. and Kumah-Abiwu, Felix},
  title = {The Political Impact of African Military Leaders: Soldiers as Intellectuals, Nationalists, Pan-Africanists, and Statesmen},
  series = {Contributions to Political Science},
  publisher = {Springer International Publishing},
  year = {2023}
}

@article{acemoglu2001extension,
  author  = {Acemoglu, Daron and Robinson, James A.},
  title   = {A Theory of Political Transitions},
  journal = {American Economic Review},
  volume  = {91},
  number  = {4},
  pages   = {938--963},
  year    = {2001}
}

@book{acemoglu2006economic,
  author    = {Acemoglu, Daron and Robinson, James A.},
  title     = {Economic Origins of Dictatorship and Democracy},
  publisher = {Cambridge University Press},
  address   = {New York},
  year      = {2006}
}

@article{acemoglu2010revolution,
  author  = {Acemoglu, Daron and Ticchi, Davide and Vindigni, Andrea},
  title   = {A Theory of Military Dictatorships},
  journal = {American Economic Journal: Macroeconomics},
  volume  = {2},
  number  = {1},
  pages   = {1--42},
  year    = {2010}
}

@article{acemoglu2008income,
  author  = {Acemoglu, Daron and Robinson, James A.},
  title   = {Persistence of Power, Elites, and Institutions},
  journal = {American Economic Review},
  volume  = {98},
  number  = {1},
  pages   = {267--293},
  year    = {2008}
}

@article{alesina1996,
  author  = {Alesina, Alberto and Özler, Sule and Roubini, Nouriel and Swagel, Phillip},
  title   = {Political Instability and Economic Growth},
  journal = {Journal of Economic Growth},
  year    = {1996},
  volume  = {1},
  number  = {2},
  pages   = {189--211},
  doi     = {10.1007/BF00138862}
}

@article{AllenEtAl2022,
  author = {Michael A. Allen and Thomas Campbell and Nicolas Hernandez and Valeryn Shepherd},
  title = {US Military Deployments and the Risk of Coup d’État},
  journal = {Foreign Policy Analysis},
  volume = {19},
  year = {2022}, 
  pages={1--20},
  doi = {10.1093/fpa/orac027}
}

@article{angeletos2007a,
  title= {Dynamic Global Games of Regime Change: Learning, Multiplicity, and the Timing of Attacks},
  author = {George-Marios Angeletos and Christian Hellwig and Alessandro Pavan},
  year = 2007,
  journal = {Econometrica},
  volume=75,
  number=3,
  pages={711-756}
}

@incollection{Bajari_Hong_Nekipelov_2013, 
  place={Cambridge}, 
  series={Econometric Society Monographs}, 
  title={Game Theory and Econometrics: A Survey of Some Recent Research}, 
  booktitle={Advances in Economics and Econometrics: Tenth World Congress}, 
  publisher={Cambridge University Press}, 
  author={Bajari, Patrick and Hong, Han and Nekipelov, Denis}, 
  editor={Acemoglu, Daron and Arellano, Manuel and Dekel, Eddie}, 
   year={2013}, pages={3–52}, 
  collection={Econometric Society Monographs}
}

@article{barro1991,
  author  = {Barro, Robert J.},
  title   = {Economic Growth in a Cross Section of Countries},
  journal = {The Quarterly Journal of Economics},
  year    = {1991},
  volume  = {106},
  number  = {2},
  pages   = {407--443},
  doi     = {10.2307/2937943}
}

@article{BelkinSchofer2003,
  author = {Aaron Belkin and Evan Schofer},
  title = {Toward a Structural Understanding of Coup Risk},
  journal = {Journal of Conflict Resolution},
  volume = {47},
  pages = {594--620},
  year = {2003}
}

@article{BjornskovRode2020,
  author = {Christian Bjørnskov and Martin Rode},
  title = {Regime Types and Regime Change: A New Dataset on Democracy, Coups, and Political Institutions},
  journal = {The Review of International Organizations},
  volume = {15},
  pages = {531--551},
  year = {2020}
}

@article{BodeaElbadawiHoule2017,
  author = {Cristina Bodea and Ibrahim Elbadawi and Christian Houle},
  title = {Do Civil Wars, Coups and Riots Have the Same Structural Determinants?},
  journal = {International Interactions},
  volume = {43},
  pages = {537--561},
  year = {2017}
}

@article{bueno22, 
  title = {Rebel motivations and repression},
  author = {Ethan Bueno de Mesquita and Mehdi Shadmehr},
  year = 2023,
  journal = {American Political Science Review},
  volume=117,
  number=2,
  pages={734--750}
}

@article{carlsson1993,
  title = {Global {{Games}} and {{Equilibrium Selection}}},
  author = {Carlsson, Hans and {van Damme}, Eric},
  year = {1993},
  journal = {Econometrica},
  volume = {61},
  number = {5},
  pages = {989}
}

@article{casper2014,
  author  = {Casper, Brett Allen and Tyson, Scott A.},
  title   = {Popular Protest and Elite Coordination in a Coup d'{\'E}tat},
  journal = {The Journal of Politics},
  year    = {2014},
  volume  = {76},
  number  = {2},
  pages   = {548--564},
  doi     = {10.1017/S0022381613001485}
}

@article{cebotari_political_2024,
	title = {Political {Fragility}: {Coups} d Etat and {Their} {Drivers}},
	volume = {2024},
	url = {https://www.elibrary.imf.org/view/journals/001/2024/034/article-A001-en.xml},
	doi = {10.5089/9798400266751.001.A001},
	number = {034},
	journal = {IMF Working Papers},
	author = {Cebotari, Aliona and Chueca-Montuenga, Enrique and Diallo, Yoro and Ma, Yunsheng and Turk, Rima A. and Xin, Weining and Zavarce, Harold},
	year = {2024},
	note = {ISBN: 9798400266751
Place: USA
Publisher: International Monetary Fund},
	pages = {A001},
}

@article{chin21,
    author = {Chin, John J and Carter, David B and Wright, Joseph G},
    title = {The Varieties of Coups D’état: Introducing the Colpus Dataset},
    journal = {International Studies Quarterly},
    volume = {65},
    number = {4},
    pages = {1040-1051},
    year = {2021}
}

@article{ChambruManeuvrierHervieu2024,
  author = {Cédric Chambru and Paul Maneuvrier-Hervieu},
  title = {Introducing HiSCoD: A New Gateway for the Study of Historical Social Conflict},
  journal = {American Political Science Review},
  volume = {118},
  pages = {1084--1091},
  year = {2024}
}

@article{ChiozzaKhalifa2024,
  author = {Giacomo Chiozza and Lena Khalifa},
  title = {The Harsh Face of the Empire by Invitation: Coups in the US World Order},
  journal = {Conflict Management and Peace Science},
  volume = {41},
  pages = {110--131},
  year = {2024}
}

@techreport{cipriani2024,
  author      = {Cipriani, Marco and Eisenbach, Thomas M. and Kovner, Anna},
  title       = {Tracing Bank Runs in Real Time},
  institution = {Federal Reserve Bank of New York},
  type        = {Staff Report},
  number      = {1104},
  year        = {2024},
  month       = {May},
  doi         = {10.59576/sr.1104}
}

@misc{cline,
 title = { Cline Center Coup d’État Project Dataset},
  author = {Peyton, Buddy and Joseph Bajjalieh and Dan Shalmon and Michael Martin and Jonathan Bonaguro and Emilio Soto},
year = {2024},
journal = {Cline Center for Advanced SocialResearch. V.2.1.3.},
publisher = {University of Illinois Urbana-Champaign},
doi = {10.13012/B2IDB-9651987_V7}
}

@misc{ClineCenter,
  author       = {{Cline Center for Advanced Social Research}},
  title        = {Cline Center for Advanced Social Research},
  year         = {2026},
  url          = {https://clinecenter.illinois.edu/},
  note         = {Accessed June 25, 2026}
}

@incollection{Darkwa2023,
  author = {Samuel Kofi Darkwa},
  title = {Military Coup D’états in Africa: A Survey},
  booktitle = {The Political Impact of African Military Leaders: Soldiers as Intellectuals, Nationalists, Pan-Africanists, and Statesmen},
  editor = {Sabella Ogbobode Abidde and Felix Kumah-Abiwu},
  publisher = {Springer International Publishing},
  year = {2023},
  pages = {1-10}
}

@article{edmond13,
title={Information Manipulation, Coordination, and Regime Change},
author={Chris Edmond},
journal={Review of Economic Studies},
year={2013},
volume={80},
number={4},
pages={1422--1458}
}

@article{feenstra2015,
  author  = {Feenstra, Robert C. and Inklaar, Robert and Timmer, Marcel P.},
  title   = {The Next Generation of the Penn World Table},
  journal = {American Economic Review},
  year    = {2015},
  volume  = {105},
  number  = {10},
  pages   = {3150--3182},
  doi     = {10.1257/aer.20130954}
}

@article{gassebner2016expect,
  author    = {Martin Gassebner and Jerg Gutmann and Stefan Voigt},
  title     = {When to Expect a Coup d'{E}tat? An Extreme Bounds Analysis of Coup Determinants},
  journal   = {Public Choice},
  year      = {2016},
  volume     = {169},
  number     = {3--4},
  pages      = {293--313},
  doi        = {10.1007/s11127-016-0365-0}
}

@article{heckman78, 
title={Dummy endogenous variables in a simultaneous equation system},
author={James Heckman},
journal={Econometrica},
year={1978},
volume={46},
number={4},
pages={931--959}
}

@article{HEINEMANN2024632,
title = {An experimental test of the global-game selection in coordination games with asymmetric players},
journal = {Journal of Economic Behavior and Organization},
volume = {218},
pages = {632-656},
year = {2024},
issn = {0167-2681},
doi = {https://doi.org/10.1016/j.jebo.2023.12.028},
url = {https://www.sciencedirect.com/science/article/pii/S0167268123004675},
author = {Frank Heinemann}
}

@article{helland21,
title={Information quality and regime change: Evidence from the lab},
author={Leif Helland and Felipe S. Iachan and Ragnar E. Juelsrud and Plamen T. Nenov},
journal={Journal of Economic Behavior and Organization},
volume={191},
year=2021,
pages={538--554}
}

@article{hollyer2015,
  title={Transparency, Protest, and Autocratic Instability},
  author={Hollyer, James R. and Rosendorff, B. Peter and Vreeland, James Raymond},
  journal={American Political Science Review},
  volume={109},
  number={4},
  pages={764--784},
  year={2015}
}

@article{jackman1978,
  author  = {Jackman, Robert W.},
  title   = {The Predictability of Coups d'Etat: A Model with African Data},
  journal = {American Political Science Review},
  year    = {1978},
  volume  = {72},
  number  = {4},
  pages   = {1262--1275}
}

@article{JenkinsKposowa1992,
  author = {J. Craig Jenkins and Augustine J. Kposowa},
  title = {The Political Origins of African Military Coups: Ethnic Competition, Military Centrality, and the Struggle Over the Postcolonial State},
  journal = {International Studies Quarterly},
  volume = {36},
  pages = {271--291},
  year = {1992}
}

@article{kocak24,
  title = {Collective procrastination and protest cycles},
  author = {Germ\'an Gieczweski and Korhan Kocak},
  year = {2024},
  journal = {American Journal of Political Science}, 
  volume={69}, 
  number={4}, 
  pages={1406-1419}
}

@article{Leon2014,
  author = {Gabriel Leon},
  title = {Loyalty for Sale? Military Spending and Coups d’État},
  journal = {Public Choice},
  volume = {159},
  pages = {363--383},
  year = {2014}
}

@article{Little2017,
  author = {Andrew T. Little},
  title = {Coordination, Learning, and Coups},
  journal = {Journal of Conflict Resolution},
  volume = {61},
  number = {1},
  pages = {204--234},
  year = {2017},
  doi = {10.1177/0022002714567953}
}

@article{Little2015,
  author = {Andrew T. Little},
  title = {Fraud and Monitoring in Non-competitive Elections},
  journal = {Political Science Research and Methods},
  volume = {3},
  number = {1},
  pages = {21--41},
  year = {2015},
  doi = {10.1017/psrm.2014.9}
}

@article{little2016,
  title={Communication Technology and Protest},
  author={Little, Andrew T.},
  journal={Journal of Politics},
  volume={78},
  number={1},
  pages={152--166},
  year={2016}
}

@article{LittleEtAl2015,
  author = {Andrew T. Little and Joshua A. Tucker and Tom LaGatta},
  title = {Elections, Protest, and Alternation of Power},
  journal = {The Journal of Politics},
  volume = {77},
  number = {4},
  pages = {1142--1156},
  year = {2015},
  doi = {10.1086/682569}
}

@article{londregan1990,
  author = {John B. Londregan and Keith T. Poole},
  title = {Poverty, the Coup Trap, and the Seizure of Executive Power},
  journal = {World Politics},
  volume = {42},
  number = {2},
  pages = {151--183},
  year = {1990},
  doi = {10.2307/2010462},
  url = {https://doi.org/10.2307/2010462}
}

@manual{marshall2020,
  title        = {Polity5: Political Regime Characteristics and Transitions, 1800--2018},
  author       = {Marshall, Monty G. and Gurr, Ted Robert},
  year         = {2020},
  organization = {Center for Systemic Peace},
  address      = {Vienna, VA},
  note         = {Dataset Users' Manual},
  url          = {https://www.systemicpeace.org/inscr/p5manualv2018.pdf}
}

@article{Masaki2016,
  author = {Takaaki Masaki},
  title = {Coups d’État and Foreign Aid},
  journal = {World Development},
  volume = {79},
  pages = {51--68},
  year = {2016}
}

@incollection{matzkin2007,
  author = {Matzkin, Rosa L.},
  title = {Nonparametric Identification},
  booktitle = {Handbook of Econometrics, Volume 6B},
  editor = {Heckman, James J. and Leamer, Edward E.},
  publisher = {Elsevier},
  pages = {5307--5368},
  year = {2007}
}

@incollection{morris2003,
  title = {Global {{Games}}: {{Theory}} and {{Applications}}},
  shorttitle = {Global {{Games}}},
  booktitle = {Advances in {{Economics}} and {{Econometrics}}: {{Proceedings}} of the {{Eight World Congress}} of the {{Econometric Society}}},
  author = {Morris, Stephen and Shin, Hyun Song},
  date = {2003},
  pages = {56},
  publisher = {{Cambridge University Press}}, 
  year=2003
}

@article{morshed23,
    author = {Morris, Stephen and Shadmehr, Mehdi},
    title = {Inspiring Regime Change},
    journal = {Journal of the European Economic Association},
    volume = {21},
    number = {6},
    pages = {2635-2681},
    year = {2023},
    issn = {1542-4766},
    doi = {10.1093/jeea/jvad023},
    url = {https://doi.org/10.1093/jeea/jvad023},
    eprint = {https://academic.oup.com/jeea/article-pdf/21/6/2635/54226976/jvad023\_teaching\_material\_morris\_and\_shadmehr.pdf},
}

@incollection{morrispal,
  title = {Global games},
  booktitle = {The New Palgrave Dictionary of Economics, Third Edition},
  author = {Stephen Morris},
  pages = {5324--5330},
  publisher = {Springer Nature}, 
  editor ={MacMillan Publishers Ltd},
  year=2018
}

@article{muller1970,
  author  = {Muller, Edward N. and Jukam, Thomas O. and Seligson, Mitchell A.},
  title   = {Diffuse Political Support and Antisystem Political Behavior: A Comparative Analysis},
  journal = {American Journal of Political Science},
  year    = {1982},
  volume  = {26},
  number  = {2},
  pages   = {240--264}
}

@article{muller1985,
  author  = {Muller, Edward N.},
  title   = {Income Inequality, Regime Repressiveness, and Political Violence},
  journal = {American Sociological Review},
  year    = {1985},
  volume  = {50},
  number  = {1},
  pages   = {47--61}
}

@article{OKane1981,
  author = {Rosemary H. T. O'Kane},
  title = {A Probabilistic Approach to the Causes of Coups d’État},
  journal = {British Journal of Political Science},
  volume = {11},
  pages = {287--308},
  year = {1981}
}

@misc{peyton2026cline,
  author       = {Peyton, Buddy and Bajjalieh, Joseph and Shalmon, Dan and Martin, Michael and Bonaguro, Jonathan and Althaus, Scott},
  title        = {{Cline Center Coup d'{E}tat Project Dataset Codebook}},
  year         = {2026},
  month        = feb,
  day          = {17},
  version      = {2.2.2},
  institution  = {Cline Center for Advanced Social Research, University of Illinois Urbana-Champaign},
  note         = {Cline Center Coup d'{E}tat Project Dataset},
  doi          = {10.13012/B2IDB-9651987\_V10}
}

@article{powell2012,
  author = {Jonathan M. Powell},
  title = {Determinants of the Attempting and Outcome of Coups d'État},
  journal = {Journal of Conflict Resolution},
  volume = {56},
  number = {6},
  pages = {1017-1040},
  year = {2012},
  doi = {10.1177/0022002712445732},
  url = {https://doi.org/10.1177/0022002712445732}
}

@article{PowellThyne2011,
  author = {Jonathan Powell and Clayton L. Thyne},
  title = {Global Instances of Coups from 1950 to 2010: A New Dataset},
  journal = {Journal of Peace Research},
  volume = {48},
  pages = {249--259},
  year = {2011}
}

@article{quinlivan99,
  author  = {Quinlivan, James T.},
  title   = {Coup-Proofing: Its Practice and Consequences in the Middle East},
  journal = {International Security},
  year    = {1999},
  volume  = {24},
  number  = {2},
  pages   = {131--165},
  publisher = {MIT Press},
  doi     = {10.1162/016228899560202}
}

@article{schroth2014,
  author  = {Schroth, Enrique and Suarez, Gustavo A. and Taylor, Lucian A.},
  title   = {Dynamic Debt Runs and Financial Fragility: Evidence from the 2007 ABCP Crisis},
  journal = {Journal of Financial Economics},
  year    = {2014},
  volume  = {112},
  number  = {2},
  pages   = {164--189},
  doi     = {10.1016/j.jfineco.2014.01.002}
}

@article{Sudduth2017,
  author = {Jun Koga Sudduth},
  title = {Coup Risk, Coup-Proofing, and Leader Survival},
  journal = {Journal of Peace Research},
  volume = {54},
  pages = {3--15},
  year = {2017}
}

@article{tamer03, 
title={Incomplete bivariate discrete response model with multiple equilibria},
author={Elie Tamer},
journal={Review of Economic Studies},
year={2003},
volume={70},
number={1},
pages={147--167}
}

@article{TysonSmith2018,
  author = {Scott A. Tyson and Alastair Smith},
  title = {Dual-Layered Coordination and Political Instability: Repression, Co-optation, and the Role of Information},
  journal = {The Journal of Politics},
  volume = {80},
  number = {1},
  pages = {45--61},
  year = {2018},
  doi = {10.1086/693986}
}

@article{vytlacil2002,
  author = {Vytlacil, Edward},
  title = {Independence, Monotonicity, and Latent Index Models: An Equivalence Result},
  journal = {Econometrica},
  volume = {70},
  number = {1},
  pages = {331--341},
  year = {2002}
}
\end{document}